\documentclass[
11pt, % The default document font size, options: 10pt, 11pt, 12pt
oneside, % Two side (alternating margins) for binding by default, uncomment to switch to one side
english, % ngerman for German
onehalfspacing, % Single line spacing, alternatives: onehalfspacing or doublespacing
headsepline, % Uncomment to get a line under the header
]{MastersDoctoralThesis} % The class file specifying the document structure

\usepackage[utf8]{inputenc} % Required for inputting international characters
\usepackage[T1]{fontenc} % Output font encoding for international characters

\usepackage{palatino} % Use the Palatino font by default

\usepackage[backend=bibtex,style=authoryear,natbib=true]{biblatex} % User the bibtex backend with the authoryear citation style (which resembles APA)
\usepackage{comment} 
\usepackage{amssymb}
\usepackage{amsmath}
\usepackage{amsthm}
\usepackage{MnSymbol}
\usepackage[frozencache]{minted}
\usepackage{multicol}
\usepackage{framed}
\usepackage{bbm}
\usepackage[multiple,para]{footmisc}
\usepackage{geometry}
\usepackage{graphicx}

\newtheorem{theorem}{Theorem}[section]
\newtheorem{corollary}{Corollary}[section]
\newtheorem{remark}{Remark}[section]
\newtheorem{definition}{Definition}[section]

\usepackage[autostyle=true]{csquotes} % Required to generate language-dependent quotes in the bibliography
\newcommand\numberthis{\addtocounter{equation}{1}\tag{\theequation}}

\newcount\colveccount
\newcommand*\colvec[1]{
        \global\colveccount#1
        \begin{pmatrix}
        \colvecnext
}
\def\colvecnext#1{
        #1
        \global\advance\colveccount-1
        \ifnum\colveccount>0
                \\
                \expandafter\colvecnext
        \else
                \end{pmatrix}
        \fi
}

\thesistitle{Large Deviation Theory for Parameter Estimation in Simple Neuron Models} % Your thesis title, this is used in the title and abstract, print it elsewhere with \ttitle
\supervisor{MSc. Johannes \textsc{Leugering} \\ Prof. Dr. Gordon \textsc{Pipa}} % Your supervisor's name, this is used in the title page, print it elsewhere with \supname
\examiner{} % Your examiner's name, this is not currently used anywhere in the template, print it elsewhere with \examname
\degree{BSc. Cognitive Science} % Your degree name, this is used in the title page and abstract, print it elsewhere with \degreename
\author{Jan H. \textsc{Kirchner}} % Your name, this is used in the title page and abstract, print it elsewhere with \authorname
\addresses{} % Your address, this is not currently used anywhere in the template, print it elsewhere with \addressname

\subject{Cognitive Science} % Your subject area, this is not currently used anywhere in the template, print it elsewhere with \subjectname
 % Keywords for your thesis, this is not currently used anywhere in the template, print it elsewhere with \keywordnames
\university{\href{https://www.uni-osnabrueck.de/en/home.html}{Universit\"at Osnabr\"uck}} % Your university's name and URL, this is used in the title page and abstract, print it elsewhere with \univname
\department{\href{http://ikw.uni-osnabrueck.de/}{Department of Cognitive Science}} % Your department's name and URL, this is used in the title page and abstract, print it elsewhere with \deptname
\group{\href{https://ikw.uni-osnabrueck.de/en/ni}{Institute of Cognitive Science, Universit\"at Osnabr\"uck}} % Your research group's name and URL, this is used in the title page, print it elsewhere with \groupname
\faculty{\href{http://faculty.university.com}{Faculty Name}} % Your faculty's name and URL, this is used in the title page and abstract, print it elsewhere with \facname

\hypersetup{pdftitle=\ttitle} % Set the PDF's title to your title
\hypersetup{pdfauthor=\authorname} % Set the PDF's author to your name
\hypersetup{pdfkeywords=\keywordnames} % Set the PDF's keywords to your keywords

\begin{document}

\frontmatter % Use roman page numbering style (i, ii, iii, iv...) for the pre-content pages

\pagestyle{plain} % Default to the plain heading style until the thesis style is called for the body content

%----------------------------------------------------------------------------------------
%	TITLE PAGE
%----------------------------------------------------------------------------------------

\begin{titlepage}
\begin{center}

\textsc{\LARGE \includegraphics[scale=0.15]{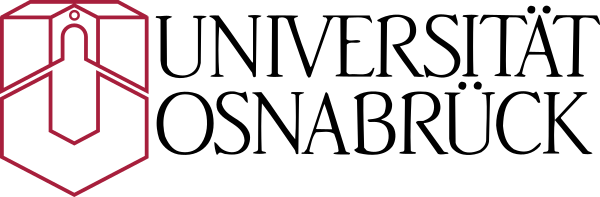}}\\[1.5cm] % University name
\textsc{\Large Bachelor Thesis}\\[0.5cm] % Thesis type

\HRule \\[0.4cm] % Horizontal line
{\huge \bfseries \ttitle}\\[0.4cm] % Thesis title
\HRule \\[1.5cm] % Horizontal line
 
\begin{minipage}{0.4\textwidth}
\begin{flushleft} \large
\emph{Author:}\\{\authorname} % Author name - remove the \href bracket to remove the link
\end{flushleft}
\end{minipage}
\begin{minipage}{0.4\textwidth}
\begin{flushright} \large
\emph{Supervisors:} \\{\supname} % Supervisor name - remove the \href bracket to remove the link  
\end{flushright}
\end{minipage}\\[3cm]
 
\large \textit{A thesis submitted in fulfillment of the requirements\\ for the degree of \degreename}\\[0.3cm] % University requirement text
\textit{in the}\\[0.4cm]
\groupname\\\deptname\\[2cm] % Research group name and department name
 
{\large \today}\\[4cm] % Date

\vfill
\end{center}
\end{titlepage}

%----------------------------------------------------------------------------------------
%	QUOTATION PAGE
%----------------------------------------------------------------------------------------

\vspace*{0.2\textheight}

\noindent\enquote{\itshape So I got a great reputation for doing integrals, only because my box of tools was different from everybody else's, and they had tried all their tools on it before giving the problem to me.}\bigbreak

\hfill Richard P. Feynman, Surely You're Joking Mr. Feynman, p. 87

%----------------------------------------------------------------------------------------
%	ABSTRACT PAGE
%----------------------------------------------------------------------------------------

\begin{abstract}
\addchaptertocentry{\abstractname} % Add the abstract to the table of contents

To investigate the complex dynamics of a biological neuron that is subject to small random perturbations we can use stochastic neuron models. While many techniques have already been developed to study properties of such models, especially the analysis of the (expected) first-passage time or (E)FPT remains difficult. In this thesis I apply the large deviation theory (LDT), which is already well-established in physics and finance, to the problem of determining the EFPT of the mean-reverting Ornstein-Uhlenbeck (OU) process. The OU process instantiates the Stochastic Leaky Integrate and Fire model and thus serves as an example of a biologically inspired mathematical neuron model. I derive several classical results using much simpler mathematics than the original publications from neuroscience and I provide a few conceivable interpretations and perspectives on these derivations. Using these results I explore some possible applications for parameter estimation and I provide an additional mathematical justification for using a Poisson process as a small-noise approximation of the full model. Finally I perform several simulations to verify these results and to reveal systematic biases of this estimator.

\end{abstract}

%----------------------------------------------------------------------------------------
%	ACKNOWLEDGEMENTS
%----------------------------------------------------------------------------------------

%\begin{acknowledgements}
%\addchaptertocentry{\acknowledgementname} % Add the acknowledgements to the table of contents

%The acknowledgments and the people to thank go here, don't forget to include your project advisor\ldots

%\end{acknowledgements}

%----------------------------------------------------------------------------------------
%	LIST OF CONTENTS/FIGURES/TABLES PAGES
%----------------------------------------------------------------------------------------

\tableofcontents % Prints the main table of contents

\mainmatter % Begin numeric (1,2,3...) page numbering

\pagestyle{thesis} % Return the page headers back to the "thesis" style

% Include the chapters of the thesis as separate files from the Chapters folder
% Uncomment the lines as you write the chapters

 % Chapter Template

\chapter{Introduction} % Main chapter title

\label{Chapter1} % Change X to a consecutive number; for referencing this chapter elsewhere, use \ref{ChapterX}

%----------------------------------------------------------------------------------------
%	SECTION 1
%----------------------------------------------------------------------------------------
``[T]he staggering complexity of even 'simple' nervous systems'' as described by \textcite{koch1999complexity} has been the subject of many research endeavors and a number of significant advances have been made. However, numerous important and very fundamental questions remain unanswered\footnote{See for example the current survey by \textcite{adolphs2015unsolved}.}. The brain is by no means a \textit{simple} system and its 100 billion neurons and hundreds of trillions of inter-neuronal connections only multiply the complexity of simple systems. Before we can understand the mechanisms of the brain as a whole, we require a better understanding of how individual \textit{neurons} process information. After understanding these constituents we might be able to understand their interactions in large systems. Mathematical \textit{models} and computer simulations have proven  to be useful tools for achieving this goal, see the book by \textcite{lytton2007computer} for an introduction to computational methods in neuroscience. By describing a biological phenomenon in the language of mathematics and then investigating the resulting model we can, if we assume that the model is expressive enough, make inferences about the original biological phenomenon. These ideas have been explored extensively, the book of \textcite{cox2006principles} provides a comprehensive survey of the related concepts. This is the approach pursued by \textit{computational neuroscience} and one of its goals is to find more useful and more expressive models of the biological reality. Mostly this reduces to finding models of the brain's main computational units: neurons \footnote{This narrow focus on neurons might be misguided. Recent empirical findings like the ones described by \textcite{fields2002new} highlight the importance of glia cells for information processing in the brain.}.

Many models of neurons have been proposed\footnote{For an overview see for example \textcite{trappenberg2009fundamentals} or again the book by \textcite{cox2006principles}.}, ranging from complete bio-physical compartment models to simple threshold point neurons. Depending on the type of model different mathematical tools are appropriate. In this thesis I propose the transfer of a tool called \textit{large deviation theory} (LDT), which is commonly used in physics and finance, to neuroscience. This theory is concerned with the exponential decay of the probability of observing stochastic deviations from a specified \textit{expected} region. Theoretically this tool is applicable to all neuron models which incorporate an additive stochastic component. It allows to analyze the small-noise probability, i.e. the probability in the situation where the influence of the stochastic component of the neuron is small, of transitioning from one state of the model to another. This probability is very interesting for making qualitative statements about the behavior of a neuron but also for the explicit simulation of a neuron, see \textcite{sacerdote2013stochastic}.

Exploring large deviation theory in the context of neuroscience can thus be interesting from two perspectives: A theoretician might apply it to gain further insights into how the quantities of interest in a neuron are connected. For the example of the \textit{Stochastic Leaky Integrate and Fire neuron model}, that I will explore more in the rest of this work, it turns out that the large deviation theory provides a perspective on how the stochastic differential equation is connected to the transition density and the first passage time via a so called \textit{rate function}. The rate function in this particular context measures squared deviation from a deterministic path. It is central to large deviation theory, see \textcite{touchette2009large}, and its usefulness in the derivation of the above quantities might hint at a similarly important interpretation in the context of simple neuron models.

The other perspective on the LDT in this context is that of an engineer who tries to incorporate a neuron model into an application like an artificial neural network. With data being available in great amounts the problem of extracting useful information from this data becomes increasingly complicated. Artificial neural networks have impressively demonstrated their usefulness for extracting structured representations from oftentimes very unstructured data, see the review paper on \textit{representation learning} by \textcite{bengio2013representation} for more details. The neuron models employed in these artificial neural networks are however often vast simplifications of a much more complicated biological reality. A constant endeavor in this field is to improve performance of such networks by implementing more sophisticated models of neurons. For more details on these \textit{spiking neuron networks} that incorporate more expressive neuron models into an artificial neural network see for example the paper by \textcite{paugam2012computing}. In this situation computational complexity becomes critical since the considered datasets are big and a small increase in computational complexity of the model size can result in a large increase in overall complexity. Small-noise approximations of the large-deviation type can again be useful here since the perturbations influencing a neuron are usually small but non-negligable and the resulting asymptotics are oftentimes much less complicated than the overall model. This approach of determining asymptotic approximations for otherwise unsolvable or unfeasible problems is called \textit{asymptotic analysis} and is explored for example in the book by \textcite{murray2012asymptotic}.

In this bachelor thesis I first provide an introduction to stochastic processes, stochastic differential equations and integrate and fire-type neuron models in chapter \ref{Chapter2} and \ref{Chapter3}. In chapter \ref{Chapter4} I provide an example application of large deviation theory to a problem from finance and I summarize some important results from LDT concerning stochastic processes defined through stochastic differential equations. In chapter \ref{Chapter5} I use the large deviation theory to investigate one of the central invariants of the stochastic leaky integrate and fire model: the expected firing time. I provide two consistent derivations for asymptotic estimators of this invariant and I explore a few possible applications of these approximations. In chapter \ref{Chapter6} I present a number of simulations I generated in \texttt{R} to test these results.
\section{Previous Work}
There is a large amount of literature available on mathematical models of neurons, see for example the classical textbook by \textcite{dayan2001theoretical}. In this work I will focus on the stochastic leaky integrate and fire model which was introduced as early as \citeyear{stein1965theoretical} by \textcite{stein1965theoretical}. A comprehensive review paper on this type of models has been published by \textcite{sacerdote2013stochastic} in which many of the classical mathematical results are summarized. Many important results on the \textit{first passage time} (FPT) and the \textit{expected first passage time} (EPFT) have been derived by \textcite{ricciardi1979ornstein}. Few people have so far actually estimated the parameters of this model, the notable exception is \textcite{lansky2006parameters}.\\
Large Deviation theory is well-established in physics and finance and besides the classical references by \textcite{richard1995overview} or \textcite{dembo2009large} there is a recent and comprehensive review paper by \textcite{touchette2009large} which also contains a rough outline of the derivation of the stationary distribution\footnote{In the \textit{Large deviations in nonequilibrium statistical mechanics} section (Example 6.3).} of an OU process that is consistent with my derivation in \ref{transition}. Within the field of neuroscience I have only found two publications applying large deviation results to neuron models. \textcite{goychuk2002ion} uses large deviation theory to analyze the opening rate of ion channels in the membrane with Kramers' law \ref{kramerslaw}. \textcite{kuehn2014large} examine stochastic neural fields with large deviation theory and discover that there are some substantial problems which arise in this setup, while also proposing possible solutions.
% Chapter 1

\chapter{Stochastic Processes \& Stochastic Differential Equations} % Main chapter title

\label{Chapter2} % For referencing the chapter elsewhere, use \ref{Chapter1} 

%----------------------------------------------------------------------------------------

% Define some commands to keep the formatting separated from the content 
\newcommand{\keyword}[1]{\textbf{#1}}
\newcommand{\tabhead}[1]{\textbf{#1}}
\newcommand{\code}[1]{\texttt{#1}}
\newcommand{\file}[1]{\texttt{\bfseries#1}}
\newcommand{\option}[1]{\texttt{\itshape#1}}

%----------------------------------------------------------------------------------------
Many interesting phenomena are to a certain degree inherently random. This randomness might stem from the underlying vagueness of the phenomenon - as for example in quantum physics where the state of a particle can only be expressed as a probability distribution - or from \textit{our} uncertainty about the phenomenon. It is for example reasonable to assume that in finance the development of the stock market is \textit{in fact} completely deterministic. However, our incomplete knowledge about all the relevant factors influencing the stock prices makes it impossible to define a wholly satisfying deterministic model, see the work by \textcite{sato1998dynamic}. Similarly a neuron's  membrane potential is the result of a big variety of interacting and inherently random cell processes. Explicit modeling of all details is again impossible due to the overwhelming quantity and complexity of relevant factors. These phenomena require the theory of probability to make adequate mathematical modeling possible. The above mentioned examples from physics and finance belong to a particular class of random phenomena that I will focus on: \textit{stochastic processes}. My exposition of these processes is going to be guided by \textcite{knill1994probability}.

\section{Stochastic Processes}
A stochastic process is a (potentially uncountably infinite) collection of random variables $\{X_{\mathrm{t}}\}$ indexed by $\mathrm{t}\in I$ on a common probability space $(\Omega,\mathcal{E},\mathbb{P})$. It is called \textit{time continuous} if $I = \mathbb{R}$ and \textit{state continuous} if $X_{\mathrm{t}} \in \mathbb{R}^n$. In this work I am going to consider stochastic processes which are time \textit{and} state continuous since these are also the processes most suited for representing biological cell processes.
A useful way to think of a stochastic process is to think of the family $\{X_{\mathrm{t}}\}$ of random variables as a family of \textit{random trajectories} indexed by $\omega\in\Omega$. For a given $\omega$ we can interpret $X_{\mathrm{t}}(\omega)$ as a sample path, i.e. a function of $\mathrm{t}$ in $\mathbb{R}^n$.

\begin{figure}
\centering
\includegraphics[scale=0.4]{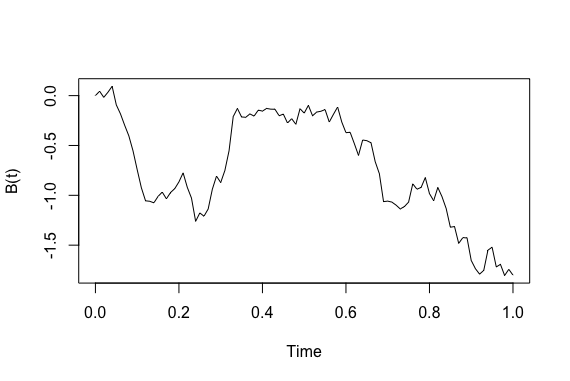}
\caption{A sample path of Brownian motion generated with \texttt{R}}
\label{BMR}
\end{figure}
Since we did not impose any constraint on the individual random variables $X_t$ this definition does not immediately provide any useful examples. It turns out however that after constructing one particular stochastic process called \textit{Brownian motion} (also known as a Wiener Process) we can characterize many interesting stochastic processes by expressing them in relation to Brownian motion. This particularly important example of a stochastic process can be characterized by the following four conditions: 
\begin{definition}\label{BMDef}
A stochastic process $B_t$ is called Brownian motion if
\begin{enumerate}
	\item $B_0(\omega) = 0, \, \forall \omega \in \Omega$
    \item $\mathbb{P}(\{\omega | B_t(\omega) \textrm{ is not continuous}\}) = 0$
    \item $t_1,\dots,t_n \in I \Rightarrow B_1 , B_2-B_1 , \dots , B_n - B_{n-1}$ are independent
    \item $0 \leq s \leq t \Rightarrow B_t - B_s \sim \mathcal{N}(0 , t-s)$
\end{enumerate}
\end{definition}
To prove that such a process exists and to achieve an intuition for its properties we can construct it as the scaling limit of a random walk. Define $X_0 = 0, X_{\mathrm{t}} = X_{{\mathrm{t}}-1} + \varepsilon_{\mathrm{t}} = \sum_{t=0}^n \varepsilon_t $ where $\varepsilon_{\mathrm{t}}$ is a Gaussian white noise process, i.e. independent, identically distributed centered random Gaussian variables. If we now set $B_t(n) := \frac{X_{\lfloor nt\rfloor}}{\sqrt[]{n}}$ then Donsker's theorem\footnote{As a special case of Donsker's invariance principle by \textcite{donsker1951invariance}.} guarantees that $\lim_{n\to\infty}B_t(n)$ exists and has properties $1-4$ of \ref{BMDef}. A sample path of a Brownian motion can be seen in \ref{BMR}. 

%----------------------------------------------------------------------------------------

\section{Stochastic Differential Equations}\label{sdeees}
After constructing Brownian Motion we are able to introduce stochastic differential equations (SDEs) which characterize a certain class of useful stochastic processes called \textit{Itô processes} as solutions to SDEs. The rigorous mathematical definition of these equations is too technical for this work so I am going to treat them only on a heuristic level. For a complete yet accessible account of the mathematical details see the classical reference for SDEs by \textcite{oksendal2003stochastic}. The basic aim is to make precise what equations of the form 
\begin{equation}
\mathrm{d} X_{\mathrm{t}} = \mu(X_{\mathrm{t}},\mathrm{t})\, \mathrm{d} \mathrm{t} +  \sigma(X_{\mathrm{t}},{\mathrm{t}})\, \mathrm{d} B_{\mathrm{t}}
\end{equation}
or equivalently written as an Itô integral
\begin{equation}
X_\mathrm{t} =\int\mu(X_{\mathrm{t}},\mathrm{t})\, \mathrm{d} \mathrm{t} +  \int\sigma(X_{\mathrm{t}},{\mathrm{t}})\, \mathrm{d} B_{\mathrm{t}}
\end{equation}
(where $B_{\mathrm{t}}$ is a Brownian Motion) mean and, if applicable, which random processes $X_t$  satisfies them. It can be shown that given relatively weak assumptions on the smoothness of the drift term $\mu$ and the diffusion term $\sigma$ the solution to such an SDE exists and sometimes can be given explicitly. This solution is then called an Itô diffusion process. Note that for the case where $\sigma = 0$ this SDE turns into a regular differential equation with a deterministic solution. So the diffusion coefficient controls the amount and the nature of the randomness added to the deterministic solution.
One way to visualize an Itô integral of the form
\begin{equation}\label{exampelIntegral}
\int \sigma(B_t,t)dB_t
\end{equation}
is to fix one $\omega\in\Omega$ and to approximate the resulting sample path of the Brownian motion $B_t(\omega)$ as a piecewise constant function on $I = \bigcupdot [t_{i},t_{i+1}[$
\begin{equation}
B_{s_i}(\omega) = c_i,\quad t_i \leq s_i < t_{i+1}
\end{equation}
This access to stochastic integrals is commonly called the Riemann-Stieltjes approach\footnote{See for example \textcite{muldowney2014understanding}.}.
$B_t(\omega)$ can be interpreted as a \textit{measure} on $I$ and then the Lebesgue integral in \ref{exampelIntegral} can be written as the sum over the constant segments
\begin{equation}
\sum_{i=0}^\infty\sigma(c_i,t_i)(c_i - c_{i-1})
\end{equation}
In this representation it also becomes obvious why we require the \textit{increments} of the Brownian motion to be Gaussian. While this is a good representation to understand the basic principle underlying stochastic integration, it is important to note that the transition from a piecewise constant sample path $B_t(\omega)$ to general Brownian Motion contains a number of substantial intricacies that impede a naive approach to stochastic integration. For the method presented in this work however the explicit solution of a SDE is not required and the naive approach should suffice. 

\section{First Exit Time \& First Passage Time}
A very important random variable that can be derived from a given stochastic process is the \textit{First Exit Time} from a given domain $D \subset \mathbb{R}^n$, i.e. the smallest time $t$ such that $X_\mathrm{t}\not\in D$: 
\begin{equation}
T^\varepsilon_D = \inf\{t \geq 0 : X_\varepsilon(t) \not\in D\}
\end{equation}
This quantity is especially interesting in the context of neuron models if the domain $D$ is chosen to be $D = ]-\infty , S]$, i.e. the first crossing of a given threshold S:
\begin{equation}
T^\varepsilon = \inf(t > 0 , X_\varepsilon(t) \geq S > v_0)
\end{equation}
The superscript $\varepsilon$ is the only parameter explicitly noted since I am going to examine these variables in the small noise limit, i.e. for small values of $\varepsilon$, and only afterwards I will analyze the resulting expression with respect to the remaining parameters. I will return to this later and remark here only that there is already a large range of literature available on how to investigate these random variables, see for example the book by \textcite{9780511606014}. 
% Chapter Template

\chapter{Stochastic Formulation Of Leaky Integrate And Fire} % Main chapter title

\label{Chapter3} % Change X to a consecutive number; for referencing this chapter elsewhere, use \ref{ChapterX}

%----------------------------------------------------------------------------------------
%	SECTION 1
%----------------------------------------------------------------------------------------
A model is a simple representation of a complex reality that helps us to understand the complex reality. To create a good model for a given application we need to identify the critical components of a system and the relevant interactions between them. Different applications require different models: If we want to understand all the details of how an action potential of a neuron is created we need an appropriately expressive model like the \textit{HH-Model} by \textcite{hodgkin1952quantitative}. 

If we are only interested in the approximated neuronal dynamics underlying the creation of an action potential, then the HH-model can be too complex to be analyzed efficiently. Especially if we are interested in imitating a neuron's behavior in an artificial neural network then computational complexity becomes critical and a simpler model is desirable. Many such \textit{minimal} models have been proposed, see for example \textcite{fitzhugh1955mathematical} and \textcite{izhikevich2003simple}.

The model that I am going to consider in this thesis is the \textit{Stochastic Version of the Leaky Integrate and Fire Neuron model} since it captures the true dynamics of a simple neuron sufficiently well while still being relatively simple. The underlying assumption of the Integrate and Fire type models is that if we are interested in a neurons spiking behavior we can limit our attention to the neuron's \textit{membrane potential} $V_\mathrm{m}$ since the evocation of an action potential critically depends on the increase of the potential above some appropriately defined \textit{threshold} $V_{th}$. My derivation of this model in the following paragraphs is based on rather heuristic arguments, for a rigorous treatment see \textcite{stevens1998novel} or \textcite{buonocore2010stochastic}.

\section{Integrate And Fire} 

Possibly the simplest models for the membrane potential of a neuron is the simple \textit{Integrate and Fire Model} which models the cell membrane as a capacitor with capacitance $C_{{\mathrm  {m}}}$ whose membrane potential $V_{\mathrm  {m}}$ can be described with respect to the input current $I(t)$ by 
\begin{equation} \label{IFmodel}
{\frac  {dV_{{\mathrm  {m}}}(t)}{dt}} =\frac{1}{C_{{\mathrm  {m}}}}I(t) , \quad V_m(0) = v_0.
\end{equation}
A constant input thus results in a linear increase of the membrane potential. Once the potential increases beyond $V_{th}$ a spike of a prefixed size $V_\mathrm{sp}$ is generated and the membrane potential is reset to its resting potential $v_r < V_\mathrm{th} < V_\mathrm{sp}$.

This model has the obvious flaw that if a temporary current is injected that does not increase the potential above the threshold then the potential stays at that value indefinitely. But in a biological neuron we would expect a decay of the potential back to its \textit{resting potential}. Therefore it is biologically more plausible to introduce a leak mechanism as described in the following section.
\section{Leaky Integrate And Fire}
In this model we assume that there is a leak current, i.e. the membrane potential decreases exponentially back to the resting potential $v_r = 0$ if the input current is zero. The conditions for a spike remain unchanged to the ones in the simple Integrate and Fire model. This relation is captured by subtracting a leak term from \ref{IFmodel} that reverts the membrane potential back to 0 in the absence of an input.
\begin{equation}\label{LIFmodel}
{\frac  {dV_{{\mathrm  {m}}}(t)}{dt}}=\frac{1}{C_{{\mathrm  {m}}}}\left(I(t)-{\frac  {V_{{\mathrm  {m}}}(t)}{R_{{\mathrm  {m}}}}}\right) , \quad V_m(0) = v_0
\end{equation}
Here $R_\mathrm{m}$ is the resistance of the membrane that controls the speed of the leaky decay. The use of Leaky Integrate and Fire (LIF) neurons instead of simple sigmoidal activation functions\footnote{Sigmoidal activation functions are of the type $\sigma(x)=\frac{1}{1 + \exp(x)}$.} can significantly improve the performance of an artificial neural network, see for example the work of \textcite{lukovsevivcius2009reservoir}. 

While this model captures many important aspects of a neuron's spiking behavior it fails to represent one important aspect: the spontaneous spiking behavior of a neuron even if the injected input remains in a sub-threshold domain. When a neuron in the resting state does not receive any input it might nonetheless eventually reach its threshold through accumulation of stochastic effects, see \textcite{alving1968spontaneous} or \textcite{hausser2004beat}. Representing this behavior in a mathematical model requires the reinterpretation of the membrane potential as a stochastic process. 
%-----------------------------------
%	SUBSECTION 1
%-----------------------------------
\section{Stochastic Formulation}\label{chap3}
\begin{figure}
\centering
\includegraphics[scale=0.4]{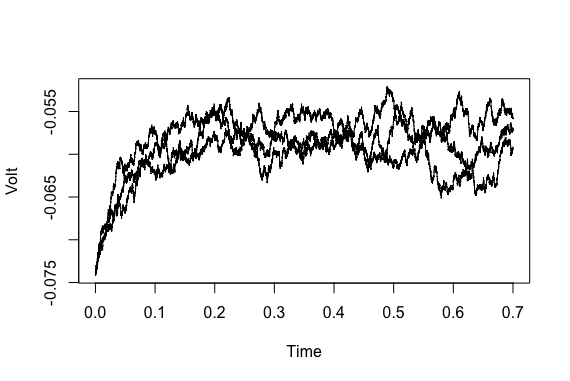}
\caption{Three sample paths from the Stochastic Leaky Integrate and Fire Model with \texttt{R} using parameter values estimated by \textcite{lansky2006parameters}: $\beta=25.8042 [1/s], v_0 = -73.92 [mV], \mu = 0.341 [V/s], \varepsilon= 0.0114[V/\sqrt[]{s}]$}
\label{OUsimu}
\end{figure}
If the additive stochastic component is assumed to be Gaussian\footnote{This is a reasonable assumption when we consider that by the central limit theorem the sum of many independent additive stochastic effects converge to a normal distribution.} we can use the methods described in section \ref{sdeees} to formulate a stochastic version of the LIF-model.
By multiplying \ref{LIFmodel} with $\mathrm{d}t$ and adding a noise term $\varepsilon\mathrm{d}W_t , \varepsilon\in\mathbb{R}_{> 0}$ to this equation we arrive at
\begin{align}
{\mathrm{d}V_{{\mathrm  {m}}}(t)}&=\frac{1}{C_{{\mathrm  {m}}}}\left(I(t)-{\frac  {V_{{\mathrm  {m}}}(t)}{R_{{\mathrm  {m}}}}}\right)\mathrm{d}t + \varepsilon \mathrm{d}W_t & \quad V_m(0) = v_0
\end{align}
By setting $\theta = \frac{1}{C_\mathrm{m}R_\mathrm{m}}$ and $\mu(t) = I(t)R_\mathrm{m}$ and by assuming $\mu(t)=\mu$ constant we arrive at the classical SDE which characterizes a centered Ornstein-Uhlenbeck process (OU process)
\begin{equation}\label{simpleform}
{\mathrm{d}V_{{\mathrm  {m}}}(t)} =-\theta\left(V_{{\mathrm  {m}}}(t) - \mu\right)\mathrm{d}t + \varepsilon \mathrm{d}W_t  \quad V_m(0) = v_0
\end{equation}
which I will use for my mathematical analysis. Alternatively\footnote{As in the paper of \textcite{ricciardi1979ornstein}.} this process can be characterized as the only \textit{stationary} and \textit{Markovian} process whose finite distributions are \textit{Multivariate Gaussians} with $\mathbb{E}(V_t) = \mu(1-\exp(-\theta t))$ and $\mathbb{V}(V_t) = \frac{\varepsilon^2}{2\theta}(1-\exp(-2\theta t))$ for a fixed t. In their work Riccardi and Sacerdote also argue for the biological plausibility of this model by deriving it from very intuitive assumptions about the biological mechanisms governing neurons and modeling dendritic input as Dirac deltas\footnote{A Dirac delta $\delta_a$ is not a function in the classical sense but instead a generalized distribution. Interpreted as a measure it becomes the Dirac point measure with $\mathbb{P}(\delta_a \in A) = \mathbbm{1}_A(a)$, i.e. 1 if and only if $a \in A$ and otherwise 0. It can be constructed as the limit of a Gaussian $\mathcal{N}(a,\sigma)$ as $\sigma \to 0$.} $I\times\delta_a(t)$.

If we instead follow the approach of \textcite{lansky2006parameters} we can also set $\mu = \mu(t)=\frac{I(t)}{C_m}$ and $\beta = \frac{1}{C_{\mathrm{m}}R_{\mathrm{m}}}$ and shift the process along the y-axis by setting $V_\mathrm{m}(t) := V_\mathrm{m}(t) - v_r$
\begin{align}\label{complicatedformulation}
\mathrm{d}V_m(t) &= \left(-\beta(V_m(t)-v_r) + \mu \right)\mathrm{d}t + \varepsilon \mathrm{d}W(t) & V_m(0) = v_0
\end{align}
to arrive at a version of the OU process with a straightforward biological interpretation: $\beta > 0$ as the spontaneous decay of the membrane potential to the resting potential $v_r$ when no external current is applied, $\mu$ as the drift coefficient which represents the external input, $\varepsilon$ as the amount of inherent stochasticity and as above an additional parameter $V_{th}$ as the threshold for spiking behavior. This version is useful for testing my results for biological plausibility. Note that we can easily transfer results derived from \ref{simpleform} to \ref{complicatedformulation} by setting $\beta := \theta$ and $\mu := v_r + \frac{\mu}{\beta}$. Three sample paths of \ref{complicatedformulation} with realistic parameter values can be seen in Figure \ref{OUsimu}.
%\section{Criticism}

There is some evidence that the membrane potential of a neuron can be modeled as an OU process but of course the model also has a number of weaknesses. Because of its simplicity many biological mechanisms are not represented, such as the refractory period after the occurrence of a spike or the fact that the criterion for the evocation of an action potential in the real neuron does not correspond to a simple threshold but rather to a separatrix in a state space, for more details see the book by \textcite{moehlis2008dynamical}. Also the additive noise is assumed to be completely state independent which is also not very realistic. These factors might be among the reasons why the OU model of neuronal activation failed to reproduce the spiking statistics of neurons in the prefrontal cortex of a monkey as was demonstrated by \textcite{shinomoto1999ornstein}. But despite all of its shortcomings the model remains one of the most thoroughly analyzed models in computational neuroscience.
% Chapter Template

\chapter{Large Deviation Theory} % Main chapter title

\label{Chapter4} % Change X to a consecutive number; for referencing this chapter elsewhere, use \ref{ChapterX}

%----------------------------------------------------------------------------------------
%	SECTION 1
%----------------------------------------------------------------------------------------

\section{Large Deviation Theory}
Large Deviation Theory (LDT) is concerned with the probability that a given random variable $X$ assumes values in a \textit{critical range}, usually values far away from the expected value of the random variable. The theory has originally been developed by the Swedish mathematician \textcite{Cramer1936} in the 1930s but many important results have since then been achieved, especially by \textcite{donsker1975asymptotic} and by \textcite{Freidlin1984}. It has since then found many interesting applications, although mainly limited to the fields of financial statistics and physics.

Most results in LDT might be considered as an extension of Cramér's \citeyear{Cramer1936} theorem which I will state here to establish the notation and the vocabulary. My exhibition will closely follow \textcite{lewis1997introduction}.
Let $I \subset \mathbb{R}$ be an interval, and $f : I \to \mathbb{R}$ a convex function; then its Legendre transform is the convex function $f^* : I^* \to \mathbb{R}$ defined by
$${\displaystyle f^{*}(x^{*})=\sup _{x\in I}(x^{*}x-f(x)),\quad x^{*}\in I^{*}} $$ The Legendre-transform can be visualized as associating to each point $f(x)$ on the graph of $f$ the negative value of the y-intercept of the tangent at $f(x)$.
For a random variable $X$ the function $M_X : I \to \mathbb{R} , t \mapsto \mathbb{E}(e^{tX})$ is called the moment-generating function of $X$ at t\footnote{Given that the integral is defined on a symmetric interval around zero.}. The name comes from the Taylor expansion around 0, $M_X(t)=\mathbb{E}(e^{tX}) = 1 + t\mathbb{E}(X) + \frac{t^2\mathbb{E}(X2)}{2!}+\dots$ from which it follows that $\mathbb{E}(X^i)=M_X^{(i)}(0)$. By taking the logarithm of $M_X(t)$ we get the cumulant generating function $C_X(t) = \log(M_X(t))$.
We can consider the Legendre transform of the cumulent generating function of $X$ and then Cramérs theorem can be stated as follows:

\begin{theorem}[Cramér's Theorem] 
Given a sequence of i.i.d. real valued random variables $X_i, i \geq 1$ with a common cumulant generating function $C_X(x) = \log\mathbb{E}[\exp(x X_1)]$ and $I(x^*) = \sup_{x\in I}(x^*x - C_X(x))$ as the Legendre transform of $C_X(x)$. Then for the empirical mean $S_n = \frac{1}{n}\sum_{i=1}^nX_i$ the following inequations hold:
\begin{samepage}
\begin{itemize}
	\item For any closed set $F \subseteq \mathbb{R}$, 
    \[\limsup_{n\to\infty} \frac{1}{n}\log \mathbb{P}(S_n \in F) \leq - \inf_{x^*\in F} I(x^*)\]
    \item For any open set $U \subseteq \mathbb{R}$, 
    \[\liminf_{n\to\infty} \frac{1}{n}\log \mathbb{P}(S_n \in U) \geq - \inf_{x^*\in U} I(x^*)\]
\end{itemize}
\end{samepage}
\end{theorem}

These bounds are \textit{tight}, i.e. they cannot be improved and represent the actual probabilities up to a normalizing factor. This can be illustrated by the following consideration: We can choose L to be an interval $[y,z] = L \subseteq \mathbb{R}$ that might also be open or half-open. Then we can always select $F$ and $U$ as the closure or the interior of the interval, i.e. $F = L^- , U = L^o$ so that the infimum of the closure and the interior over the rate function coincide: $\inf_{x^*\in L^-}I(x^*)=\inf_{x^*\in L^o}I(x^*)$. Under this condition it can also been shown\footnote{See Section 3 of \textcite{richard1995overview}.} that the $\limsup$ and the $\liminf$ coincide and we arrive at:
\begin{equation}
\lim_{n\to\infty} \frac{1}{n}\log \mathbb{P}(S_n \in L) = - \inf_{x\in L} I(x)
\end{equation}
or equivalently for $n >> 0$
\begin{equation}
\mathbb{P}(S_n \in L) \approx \exp\{- n\inf_{x\in L} I(x)\}
\end{equation}
which gives us an explicit expression for the distribution of the mean.\footnote{\textcite{donsker1976asymptotic} have discovered that the Legendre transform of a cumulant generating function evaluated at $a$ equals the minimal \textit{Kullback–Leibler divergence} between the original distribution and a suitably chosen measure dependent on $a$. See Section 4 of  \textcite{fischerlarge} for a precise statement.}

I will apply Cramér's Theorem to a classical situation that has already many similarities with the situation relevant for the current work: minimizing the risk of having costs accumulate above a given income. Consider a secure steady income of $p\in \mathbb{R}$ units and random i.i.d. daily payments described by the random variables $\{X_i\}$. The question is how probable it is that the accumulated costs $\sum_{i=1}^TX_i$ over a period of $T$ days exceed the secure income $pT$, i.e. the probability given by $\mathbb{P}(\sum_{i=1}^TX_i \geq pT) = \mathbb{P}(\frac{1}{T}\sum_{i=1}^TX_i \geq p)$. This expression allows the application of Cramérs Theorem with a suitable rate function $I(x)$ for the random variables $\{X_i\}$
\begin{equation}
\mathbb{P}(\frac{1}{T}\sum_{i=1}^TX_i \geq p) \approx \exp\{- T \inf_{x\in A_p} I(x)\}
\end{equation}
where $A_p = \{x | x \geq p\}$ is a half open interval and $T$ has to be sufficiently large. 
To evaluate this expression we need to make further assumptions on the random variables describing the daily payments. For this example we may assume that they are independent and identically normal distributed with mean $\mu$ and variance $\sigma^2$. Then it is well-known\footnote{One way to derive this is to realize that $M_X(\theta) = \mathcal{F}^{-1}(X)(-i\theta)$, where $\mathcal{F}^{-1}$ denotes the inverse Fourier transform, and using that for $X \sim \mathcal{N}(\mu,\sigma)$ the inverse Fourier transform is given by $\mathcal{F}^{-1}(X)(t) = e^{-\frac{1}{2}\sigma^2t^2-i\mu t}$. Alternatively one can also simply solve the integral $\mathbb{E}(e^{\theta X})=\int_{-\infty}^\infty e^{\theta x}\frac{1}{\sqrt[]{2\pi}\sigma}e^{-\frac{(x-\mu)^2}{2\sigma^2}}\mathrm{d}x$ by completing the square in the exponential.} that the corresponding moment generating function is $M(\theta) = e^{\theta \mu +{\frac {1}{2}}\sigma ^{2}\theta^{2}}$ and the cumulant generating function is therefore $\log M(\theta) = \theta \mu +{\frac {1}{2}}\sigma ^{2}\theta^{2}$. So $I(x)$ can be calculated as
\begin{align}
I(x) &= \sup _{\theta\in I}(x\theta-\log M(\theta)) = \sup _{\theta\in I}\left(x\theta- \left(\theta \mu +{\frac {1}{2}}\sigma ^{2}\theta^{2}\right)\right)\\
&= x\left(\frac{x-y}{\sigma^2}\right) - \mu\left(\frac{x-\mu}{\sigma^2}\right) -\frac{1}{2}\sigma^2 \left(\frac{x-\mu}{\sigma^2}\right)^2 \\ 
&= \frac{1}{2}\left(\frac{x-\mu}{\sigma}\right)^2
\end{align}
since the supremum over $\theta$ is given by $\theta^* = \frac{x-\mu}{\sigma}$ which can be seen by taking the derivative and setting it equal to zero.
If we fix some small probability $\exp\{-r\}$ we can calculate the corresponding value for p so that the asymptotic probability equals this small probability:
\begin{align}
\exp\{-r\} &= \exp\{- T \inf_{x\in A_p} I(x)\} \\
\Leftrightarrow \frac{r}{T} &= \inf_{x\in A_p} I(x) = I(p) = \frac{1}{2}\left(\frac{p-\mu}{\sigma}\right)^2 \\
\Rightarrow p &= \mu + \sigma\sqrt[]{\frac{2r}{T}}
\end{align}
Since the information function is convex with a unique minimum\footnote{This is one of the characteristic features of a rate function that can can be seen immediately in this case.} at $\mu$ it assumes its infimum over $A_p$ at the value from $A_p$ that is closest to $\mu$. Given the reasonable assumption\footnote{The daily income should be higher than the average daily loss, otherwise it's not reasonable to expect that losses are rare.} that $p\geq\mu$ the second step from the calculation follows. The second solution from taking the square-root in the third line is dropped for the same reason. From these calculations it follows that to ensure a small probability of a negative balance the daily income should be bigger than the expected loss by $\sigma\sqrt[]{\frac{2r}{T}}$.

%-----------------------------------
%	SUBSECTION 1
%-----------------------------------
\section{Freidlin–Wentzell theorem}\label{fwt}
We are going to turn to what is often called \textit{level-2 large deviations} as coined by \textcite{book:55349} in \textit{Entropy, Large Deviations, and Statistical Mechanics}. While the calculations from the previous section can be considered \textit{standard}, the following theorem by \textcite{Freidlin1984} has found much less application outside of physics and finance. For previous applications of the theorem in physics see for example \textcite{luchinsky1998analogue} or \textcite{Landa20001} and for a recent application in finance see  \textcite{pham2010large}. One publication from the field of neuroscience that is concerned with Nonlocal Stochastic Neural Fields is \textcite{kuehn2014large}.

Let $X_\varepsilon\in\mathbb{R}^n$ be a stochastic process that satisfies $$dX_\varepsilon(t)=b(X_\varepsilon)\mathrm{d}t+ \sqrt{\varepsilon}\sigma(X_\varepsilon)\mathrm{d}W(t),\quad X_\varepsilon(0) = x_0 , t\in[T_1,T_2]$$ for uniformly Lipschitz\footnote{A function is called Lipschitz (continuous) if there exists a constant $L\in\mathbb{R}$ such that for all $x_1,x_2\in\mathbb{R}^n: ||f(x_1)-f(x_2)|| \leq L|x_1-x_2|$. This implies usual continuity.} drift function $b: \mathbb{R}^n \to \mathbb{R}^n$ and diffusion matrix $\sigma: \mathbb{R}^n\to\mathbb{R}^{n\times n}$. It can be shown that the probability of large deviations from the deterministic solution,
\begin{equation}
\mathbb{P}(\sup_{t}|X_\varepsilon(t) - x(t)| > \delta) , \delta > 0
\end{equation}
i.e. from the function $x(t)$ which satisfies 
\begin{equation}\label{detersys}
	\frac{\mathrm{d}x(t)}{\mathrm{d}t} = b(x(t)), x(0)=x_0 
\end{equation} equals zero in the small noise limit $\varepsilon\to 0$. Nonetheless large deviations due to accumulation of noise \textit{can} still occur even though they become increasingly unlikely. The general result of Freidlin and Wentzell shows that the decay of this probability in the small-noise limit is exponential and they provide an explicit expression for the exponential. It applies to stochastic processes on almost arbitrary probability spaces but for this work I only consider real-valued stochastic processes as formulated above. Applied to this case the theorem can be stated as in \textcite{peithmann2007large} where a full proof of a generalized statement can also be found:
\begin{theorem}[Freidlin-Wentzell theorem]\label{fwtheo}
Let $X_\varepsilon, \varepsilon > 0$ be the family of $\mathbb{R}^n$-valued processes defined by
\begin{equation}\label{statement421}
\mathrm{d}X_\varepsilon(t) = b(X_\varepsilon(t)) \mathrm{d}t +\sqrt[]{\varepsilon}\sigma(X_\varepsilon(t)) \mathrm{d}W(t), X_\varepsilon(0) = x_0 \in \mathbb{R}^n 
\end{equation}
on a fixed time interval $[T_1, T_2]$, where $b$ and $\sigma$ are Lipschitz continuous, and $W$ is
an n-dimensional Brownian motion. Let $C[T_1,T_2]$ be the space of absolutely continuous functions in $\mathbb{R}^n$ with square integrable derivatives defined on $[T_1,T_2]$. If $\sigma$ is invertible and $a = \sigma\sigma^T$ is uniformly positive definite\footnote{A matrix $M\in\mathbb{R}^{n\times n}$ is called positive definite if for all $z\in\mathbb{R}^n, z^TMz>0$. It is called \textit{uniformly} positive definite if $M(x)$ is a function of a vector $x\in\mathbb{R}^n$ and for a given $z\in\mathbb{R}^n$ there exists an $\varepsilon>0$ such that for all $x\in\mathbb{R}^n, z^TM(x)z \geq \varepsilon > 0$.} then
\begin{itemize}
	\item For any closed subset $F \subseteq C[T_1,T_2]$, 
    \[\limsup_{\varepsilon\to 0} \varepsilon\log \mathbf{P}(X_\varepsilon \in F) \leq - \inf_{\chi\in F} J(x)\]
    \item For any open subset $U \subseteq C[T_1,T_2]$, 
    \[\limsup_{\varepsilon\to 0} \varepsilon\log \mathbf{P}(X_\varepsilon \in U) \geq - \inf_{\chi\in U} J(x)\]
\end{itemize}
where the rate function $J[\chi] = \frac{1}{2}\int_{T_1}^{T_2} [\dot{\chi} - b(\chi)]^Ta^{-1}(\chi)[\dot{\chi} - b(\chi)] \mathrm{d}t$ and $\mathbf{P}$ denotes the path density. 

This is commonly written as
\begin{align}\label{gibbsform}
\mathbf{P}_{\varepsilon}[\chi] \asymp \exp\{-\frac{\inf_{\chi} J[\chi]}{\varepsilon}\}
\end{align}
\end{theorem}
As above, when $F$ and $U$ are chosen appropriately this asymptotic equality $\asymp$ corresponds to approximate equality $\approx$.

Written in the form of \ref{gibbsform} the density looks very similar to the \textit{Gibbs measure} of a random process given by $\frac{1}{Z(\varepsilon)}e^{-E(\chi)/\varepsilon}$ where $E(\chi)$ is called the \textit{energy} of the path $\chi$. This measure is commonly used in statistical mechanics and thermodynamics but it has also found prominent application in machine learning as \textit{Gibbs sampling} in restricted Boltzmann machines. The existence of such an energy function for a stochastic process defined by \ref{statement421} is guaranteed by the theorem of Hammersley and \textcite{clifford1990markov} since all Itô diffusion processes have the Markov property\footnote{A proof can be found in section 7 of \textcite{oksendal2003stochastic}.}. The Freidlin-Wentzell theorem then provides an expression for the energy in the small-noise limit and says that the density is determined by the path minimizing this energy.

In physics the \textit{rate function} $J[\chi]$ is commonly called the \textit{action of a system}. In the case of \ref{fwtheo} it would be the action with respect to a \textit{Lagrangian} functional defined by $\mathcal{L}(\chi) = \frac{1}{2}[\dot{\chi} - b(\chi)]^Ta^{-1}(\chi)[\dot{\chi} - b(\chi)]$. The infimum path over $\chi\in C[T_1,T_2]$ of such a system is commonly called a stationary path and by Hamilton's principle it is the path through the state spate that the deterministic system corresponding to $\mathcal{L}$ would take. The connection between stochastic processes and a Lagrangian has been examined by \textcite{de1992stochastic} although the particular Lagrangian $\mathcal{L}$ that is implied in \ref{fwtheo} is not mentioned. I am not going to consider this connection to physics further and I will focus on the application of the theorem to simple neuron models.

\section{Kramers' law}\label{kl}
The infimum over the rate function that is determined for applying the Freidlin-Wentzell theorem can be used to derive another very important invariant: the expected value of the first-exit time $T^\varepsilon_D$. This result is called Kramers' law\footnote{Not to be confused with Harald Cramér.}, named after H.A. Kramers who determined an approximate equation for the diffusion of Brownian Motion in chemical reactions, and can be found in its LDT formulation alongside a full proof in \textcite{Freidlin1984}. $x(t)$ denotes again the deterministic solution of \ref{detersys} and $\mathbb{E}^{x_0}$ denotes that the expected value is calculated with respect to the stochastic process that satisfies the initial condition $X(0) = x_0$.
\begin{theorem}[Kramers' law]\label{kramerslaw}
Let $D\subset\mathbb{R}^n$ be a bound set that is enclosed in the domain of attraction of  the system \ref{detersys}, i.e. that satisfies:
\begin{enumerate}
\item The deterministic system possesses a unique stable equilibrium point $x^* \in D$
\item The solutions of the deterministic system satisfy
\begin{equation}
x_0 \in D \Rightarrow x_t \in D\; \forall t > 0 \wedge \lim_{t\to\infty}x_t = x^*.
\end{equation}
\end{enumerate}
Then assuming that for $F_{t,y} = \{\chi\in C[0,t] : \chi(0) = x^* , \chi(t) = y\}$ $$V^* := \inf_{y\in\delta D}\inf_{t>0}\inf_{\chi\in F_{t,y}}J(\chi) < \infty$$
it follows that for all initial conditions $x_0 \in D$
\begin{equation}\label{rhskr}
\lim_{\varepsilon\to0}\varepsilon \log \mathbb{E}^{x_0}[T_D^\varepsilon]= V^*
\end{equation}
and
\begin{equation}
\lim_{\varepsilon\to 0}\mathbb{P}(e^{\frac{V^*-\delta}{\varepsilon}} < T_D^\varepsilon < e^{\frac{V^* + \delta}{\varepsilon}}) = 1 \textrm{ for all } \delta > 0
\end{equation}
\end{theorem}
\begin{figure}
\centering
\includegraphics[scale=0.4]{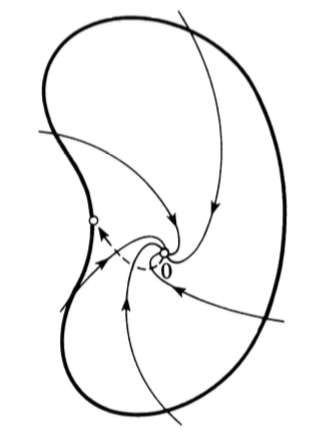}
\caption{An illustration of the exit from a bounded domain which contains an attractor, taken from \textcite{Freidlin1984} page 105.}
\label{asympExit}
\end{figure}

The disappearance of $x_0$ on the right hand side of the mean first-passage time (FPT) in \ref{rhskr} is due to the fact that ``[in the] asymptotic case the mean FPT, $\mathbb{E}[T]$ loses the dependency upon the initial value [$x_0$]``\footnote{Section 5.4.1.3 in Stochastic Biomathematical Models: with Applications to Neuronal Modeling by \textcite{StoBioMod}}\footnote{See also Section 2 of \textcite{berglund2011kramers} for more intuition and for the sketch of a proof.}. When $\varepsilon$ is small enough the time required to leave the domain of attraction is by orders of magnitude greater than the time to reach the attractor from any given $x_0 \in D$ so that by the Markovian nature of the Itô diffusion processes any initial deviation becomes irrelevant. This simplifying assumption gives us an easy way of estimating the expected first-passage time but it also indicates how rough such an estimate will be. An illustration of this process can be seen in Figure \ref{asympExit}. While the exponential $e^{V^*/\varepsilon}$ will dominate in the small-noise limit, subexponential factors can strongly bias this estimate. Methods to determine this subexponential factor have been proposed by \textcite{PhysRevE.48.931} but remain according to \textcite{berglund2011kramers} and my research currently without rigorous mathematical proof so I have not incorporated them into this work. 
% Chapter Template

\chapter{Parameter Estimation} % Main chapter title

\label{Chapter5} % Change X to a consecutive number; for referencing this chapter elsewhere, use \ref{ChapterX}

%----------------------------------------------------------------------------------------
%	SECTION 1
%----------------------------------------------------------------------------------------

Generally, the more details of a complex reality a model incorporates, the more difficult it becomes to investigate the properties of the model. While the stochastic version of the Leaky Integrate and Fire Neuron model can be stated rather concisely, the identification of important densities and related invariants soon becomes very difficult. A good example for this is the work of \textcite{10.2307/3214232} who was the first to solve the long standing problem of finding a closed-form expression for the moments of the first-passage-time of an OU-process. His publication marks the end of an elongated effort of many researchers to achieve detailed knowledge about the behavior of the OU-process and its corresponding random first-passage time $T^\varepsilon$.

The results by Ricciardi are complete from a mathematical standpoint, however their derivation is very complicated in places and obstructs a clear understanding of what these results \textit{mean} for the neuron model that they describe. Since 1988 some simpler derivations of related results have been presented\footnote{For example by \textcite{yi2006useful}.} but they are still very technical.

In this chapter I am going to apply the previously introduced methods from large deviation theory to derive well-known results from neuroscience using much simpler mathematics than the original publications. In exchange the number and the scope of my derived results is also much smaller than the original publications. These methods should be applicable since the noise term $\varepsilon$ that was empirically determined by \textcite{lansky2006parameters} or \textcite{picchini2008parameters} is on the order of $10^{-2} [V/\sqrt[]{s}]$ which is small enough that large deviations become rare events.

\section{Stochastic Leaky Integrate and Fire}
As introduced in \ref{chap3} the Stochastic LIF is a neuron model that is described through a stochastic version of the usual LIF differential equation. While in the deterministic model a sub-threshold excitation of the neuron will never evoke an action potential, once stochastic fluctuations are introduced into the system there is always a \textit{probability} that a potential is evoked.

For the OU-process the transition probability that the membrane potential lies in $[x -\delta , x + \delta]$ at time $T_2$ given that at time $T_1$ it was equal to $v_1$ can be stated explicitly through the corresponding density function $\mathbb{P}(V_\mathrm{m}(T_2)\in[x-\delta,x+\delta]|V_\mathrm{m}(T_1)=v_1)\approx \mathbf{P}(x|T_2,v_1,T_1)\times 2\delta$. Large deviation theory tells us that this density decays exponentially fast with rate $\inf_\chi J[\chi]/\varepsilon^2$ where the infimum is taken over all the possible (continuous) paths from $v_0$ to $x$. In the one-dimensional case where the noise is state-independent the rate function $J[\chi]$ reduces to $\int_{T_1}^{T_2} (\dot{\chi}(t) - b(\chi)(t))^2 \mathrm{d}t$ which is minimized by the path whose derivative diverges the least (in a mean-squared sense) from the derivative of the deterministic solution given by $\frac{\mathrm{d}V(t)}{\mathrm{d}t}=\mathbf{b}(V)(t)$. Thus in the small-noise limit the probability of an unlikely deviation from the deterministic trajectory is determined by the most likely of possible (unlikely) trajectories consistent with the model.
\begin{theorem}\label{transition}
Let $V_{{\mathrm  {m}}}(t)$ be an Ornstein-Uhlenbeck process satisfying
\begin{equation}\label{sde}
\mathrm{d}V_m(t) = -\theta(V_{{\mathrm  {m}}}(t) - \mu)\mathrm{d}t + \varepsilon \mathrm{d}W(t) \quad V_m(T_1) = v_1\end{equation}
Then in the small noise limit $\varepsilon \to 0$ the transition density $\mathbf{P}(x|T_2,v_1,T_1)$ approaches the density of a Gaussian Normal distribution with the following parameters $\mathcal{N}(\mu - (\mu - v_1)e^{\theta(T_1 - T_2)},\frac{\varepsilon^2}{2\theta} \left(1 - \mathrm{e}^{2\,\theta (\mathrm{T_1-T_2})} \right)$.
\end{theorem}
\begin{comment}
\begin{theorem}\label{transition}
Let $V_{{\mathrm  {m}}}(t)$ be an Ohrnstein-Uhlenbeck process satisfying
\begin{equation}\label{sde}
\mathrm{d}V_m(t) = \mathbf{b}(V_{{\mathrm  {m}}}(t))\mathrm{d}t + \sqrt{\varepsilon^2} \mathrm{d}W(t) \quad V_m(T_1) = v_0\end{equation}
with $\mathbf{b}(V_{{\mathrm  {m}}}(t)) = \left(-\beta(V_m(t)-v_0) + \mu \right)$. Then in the small noise limit $\varepsilon^2 \to 0$ the transition density $\mathbf{P}(x|T_2,v_0,T_1)$ approaches the density of a Gaussian Normal distribution with the following parameters $\mathcal{N}(v_0 + \frac{\mu}{\beta} - \frac{\mu}{\beta}e^{\beta(T_1-T_2)},\frac{\varepsilon^2}{2\beta} \left(1 - \mathrm{e}^{2\,\beta (\mathrm{T_1-T_2})} \right)$.
\end{theorem}
\end{comment}

%-----------------------------------
%	SUBSECTION 1
%-----------------------------------
\begin{proof}
Let $E_{(T_1,T_2,U)}=\{\omega \in C^1[T_1 , T_2]\ | \omega(T_1) = v_1 , \omega(T_2) = U \}$. With \ref{fwt} applied to \ref{sde} we immediately derive $$\mathbf{P}(V_\mathrm{m} \in E_{(T_1,T_2,U)}) \asymp \exp\{-\inf_{\omega\in E_{(T_1,T_2,U)}}(J[\omega]) / \varepsilon^2\}$$ with $J:E_{(T_1,T_2,U)}\to \mathbb{R} , \omega \to \frac{1}{2}\int_{T_1}^{T_2} [\dot{\omega} +\theta(\omega-\mu)]^2\mathrm{d}t$ since $a = 1$ and  $\mathbf{b}(V_{{\mathrm  {m}}}(t)) = -\theta(V_m(t)-\mu)$. A necessary condition for $\omega^*$ to be an extremum of $J$ is that $\omega^*$ satisfies the Euler-Lagrange equation\footnote{For more details on stationary paths and the EL-equation see \textcite{sasane2016optimization}.}
\begin{equation}
\frac{\delta{F}}{\delta{A}}(x(t), \dot{x}(t), t) - \frac{\mathrm{d}}{\mathrm{d}t}\left(\frac{\delta F}{\delta B}(x(t), \dot{x}(t), t)\right)= 0,\quad t \in [T_1, T_2]
\end{equation}
where $F(A,B,C) = [B +\theta(A-\mu)]^2$ so that $J[\omega] = \frac{1}{2}\int_{T_1}^{T_2}F(\omega(t),\dot{\omega}(t),t)\mathrm{d}t$. Plugging in F we get the autonomous second order differential equation
\begin{align}
\theta[\dot{\omega} + \theta(\omega - \mu)] - [\ddot{\omega} + \theta\dot{\omega}] &= 0 \\
\Leftrightarrow\qquad \theta^2\omega - \theta^2\mu &= \ddot{\omega}
\end{align}
which is solved by
\begin{multline}
\omega^*(t) = \mu + \frac{e^{-\theta t}\left(Ue^{T_1\theta} - \mu e^{T_1\theta} + \mu e^{T_2\theta} - v_1e^{T_2\theta}\right)}{e^{\theta(T_1-T_2)} - e^{\theta(T_2-T_1)}} \\
- \frac{e^{\theta t}\left(Ue^{-T_1\theta} - \mu e^{-T_1\theta} + \mu e^{-T_2\theta} - v_1e^{-T_2\theta}\right)}{e^{\theta(T_1-T_2)} - e^{\theta(T_2-T_1)}}
\end{multline}
with $w^*(T_1)=v_1 , w^*(T_2)=U$ as required. 
This extremum is also a minimum\footnote{Sufficient conditions for a minimum are in general very difficult to obtain in the calculus of variations, see the historical review by \textcite{fraser2009sufficient} for an overview of methods.} since for $\omega_1,\omega_2 \in C^1[T_1,T_2]$ and $\lambda \in [0,1]$ we can write $J[\lambda\omega_1 + (1-\lambda)\omega_2]$ as
\begin{align}
	J[\lambda\omega_1 + (1-\lambda)\omega_2] &= \frac{1}{2}\int_{T_1}^{T_2} \left[ \lambda \dot{\omega}_1 + (1-\lambda)\dot{\omega}_2 + \theta(\lambda \omega_1 + (1-\lambda)\omega_2 - \mu)  \right]^2\mathrm{d}t \\
	&= \frac{1}{2}\int_{T_1}^{T_2} \left[ \lambda( \dot{\omega}_1 + \theta(\omega_1 - \mu)) + (1-\lambda)(\dot{\omega}_2 + \theta(\omega_2 - \mu)) \right]^2\mathrm{d}t \\
	& \leq   \lambda J[\omega_1] + (1-\lambda)J[\omega_2]
\end{align} 
\begin{comment}
	&= \frac{1}{2}||\lambda(\dot{\omega}_1 + \theta(\omega_1 - \mu)) + (1-\lambda)(\dot{\omega}_2 + \theta(\omega_2 - \mu))  ||_2^2 \\
	& \stackrel{\Delta}{\leq} \lambda^2 J[\omega_1] + (\lambda - 1)^2J[\omega_2] \\
	& \leq   \lambda J[\omega_1] + (1-\lambda)J[\omega_2]
where $\stackrel{\Delta}{\leq}$ denotes the minkowski inequality for the 2-norm and the last inequality follows because $J[\omega]\geq 0, \forall \omega\in C^1[T_1,T_2]$ and $\lambda \in [0,1]$
\end{comment}  
where the last inequality follows since the composition of taking the $L^2$ norm and squaring it is convex\footnote{The composition $\phi \circ \psi$ of two convex functions $\phi,\psi$ is convex if $\phi$ is also monotone. This is given since $x \mapsto x*x$ restricted to $\mathbb{R}_{\geq}$ is monotonically increasing and the $L^2$ norm only assumes values in $\mathbb{R}_{\geq}$.}. Therefore J is a convex functional and the local extremum $\omega^*$ must also be a global minimum.
We can derive 
\begin{align*}
\inf_{\omega\in E_{T_1,T_2,U}}(J[\omega]) &= J[\omega^*] = \frac{1}{2}\int_{T_1}^{T_2} [\dot{\omega}^* +\theta(\omega^*-\mu)]^2\mathrm{d}t \numberthis\\
 &= \frac{1}{2}\int_{T_1}^{T_2}\bigg[ \frac{2\theta e^{\theta t}\left(Ue^{-T_1\theta} - \mu e^{-T_1\theta} + \mu e^{-T_2\theta} - v_1e^{-T_2\theta}\right)}{e^{\theta(T_1-T_2)} - e^{\theta(T_2-T_1)}}\bigg]^2 \mathrm{d}t\numberthis\\
&= \int_{T_1}^{T_2}2\theta e^{2\theta t}\mathrm{d}t\times\\ &\qquad\theta\bigg[ \left(\frac{e^{\theta(T_1 + T_2)}}{e^{\theta(T_1 + T_2)}}\right)\frac{Ue^{-T_1\theta} - \mu e^{-T_1\theta} + \mu e^{-T_2\theta} - v_1e^{-T_2\theta}}{e^{\theta(T_1-T_2)} - e^{\theta(T_2-T_1)}}\bigg]^2 \numberthis\\
&= (e^{2\theta T_2} - e^{2\theta T_1})\times\theta\bigg[ \frac{Ue^{T_2\theta} - \mu e^{T_2\theta} + \mu e^{T_1\theta} - v_1e^{T_1\theta}}{e^{2\theta T_1} - e^{2\theta T_2}}\bigg]^2 \numberthis\\
&= \frac{e^{2\theta T_2}\left(U - \mu + (\mu - v_1)e^{\theta(T_1-T_2)}\right)^2}{\frac{1}{\theta}(e^{2\theta T_2} - e^{2\theta T_1})}\numberthis\\
&= \frac{\left(U - \mu + (\mu - v_1)e^{\theta(T_1-T_2)}\right)^2}{\frac{1}{\theta}(1- e^{2\theta (T_1-T_2)})} = \frac{(U - \xi)^2}{2\sigma^2}\numberthis
\end{align*}
with $\xi = \mu - (\mu - v_1)e^{\theta(T_1-T_2)}$ and $\sigma^2 = \frac{1}{2\theta} \left(1 - \mathrm{e}^{2\,\theta (\mathrm{T_1-T_2})}\right)$

Therefore the original statement follows since the resulting expression must be a probability density in U
\begin{align}
\mathbf{P}(V_\mathrm{m} \in E_{(T_1,T_2,U)}) & \asymp \exp\{-\frac{(U - \xi)^2}{2\varepsilon^2\sigma^2}\}
\end{align}
\end{proof}

\begin{corollary}
For $T_1 \to -\infty$ we retrieve the stationary normal distribution of the OU-process with $\mathcal{N}(\mu , \frac{\varepsilon^2}{2\theta})$.
\end{corollary}

\begin{corollary}
For $T_1 = 0$ we get the transition density as a normal distribution with $\mathcal{N}(\mu - (\mu - v_1)e^{-\theta T_2}, \frac{\varepsilon^2}{2\theta}(1-e^{-2\theta T_2}))$. If we set $\beta := \theta, v_r = v_0$ and $\mu := v_r + \frac{\mu}{\beta}$ as described in \ref{chap3} we retrieve the transition density as it is stated in \textcite{lansky2006parameters} with $\mathcal{N}(v_1 + \frac{\mu}{\beta}(1-e^{-\beta T_2}) , \frac{\varepsilon^2}{2\beta}(1-e^{-2\beta T_2}))$
\end{corollary}
Next I apply Kramers' law to get an estimate for the expected first-passage time, i.e. the time the neuron is expected to cross a given threshold $V_\mathrm{th}$ given that at time 0 the membrane potential equals $v_1$. This estimate is again derived by minimizing the rate function: We determine the most likely path going from $v_r$ at time $0$ to $V_\mathrm{th}$ in an arbitrary amount of time. This can again be considered as choosing the most likely of a number of unlikely paths. Kramers' law tells us that this path determines (in the small-noise-limit) the expected value of the first-passage time $T^\varepsilon$. 
%The disappearance of $v_1$ is due to the fact that ``[in the] asymptotic case the mean FPT, $\mathbb{E}[T]$ loses the dependency upon the initial value [$v_1$]``\footnote{Section 5.4.1.3 in Stochastic Biomathematical Models: with Applications to Neuronal Modeling by \textcite{StoBioMod}}. When $\varepsilon$ is small enough the time required to leave the domain of attraction is by orders of magnitude greater than the time to reach the attractor from any given $v_1 \in D$ so that by the Markovian nature of the Itô diffusion processes any initial deviation becomes irrelevant.

\begin{theorem}\label{kramer_est}
For $D = ]x , z[$, $T_1 = 0 , x < v_1 < \mu < z$ the expected value of the first exit time $T_D^\varepsilon$ from $D$ of the process $V_\mathrm{m}(t)$ with $V_\mathrm{m}(0)=v_1$ can be approximated for small values of $\varepsilon$ by
\begin{equation}
\mathbb{E}^{v_1}(T_D^\varepsilon) \approx \begin{cases}
e^{\frac{\theta}{\varepsilon}(z-\mu)^2} & \mathrm{if } |x-\mu| \geq |z-\mu|\\
e^{\frac{\theta}{\varepsilon}(x-\mu)^2} & \mathrm{else}\\
\end{cases}
\end{equation}
and also
\begin{equation}
\quad \lim_{\varepsilon\to 0}\mathbb{P}(\mathbb{E}^{v_1}(T_D^\varepsilon)e^{-\frac{\delta}{\epsilon}} < T_D^\varepsilon < \mathbb{E}^{v_1}(T_D^\varepsilon)e^{\frac{\delta}{\epsilon}}) = 1 \textrm{ for all } \delta > 0
\end{equation}

\end{theorem}
\begin{proof}
The deterministic system
\begin{equation}
\frac{\mathrm{d}v(t)}{\mathrm{d}t} = -\theta(v(t) - \mu) , \quad v(0) = v_1
\end{equation}
has a single global point attractor $v^*$ given by
\begin{align}\label{attractor}
-\theta(v^* - \mu) &= 0 &\Rightarrow v^* &= \mu
\end{align}
%which fulfills the necessary conditions of Kramers' Law (\ref{kramerslaw}). From the proof of \ref{transition} we know that for $F_{t,y} = \{\chi : \chi(0) = v^* , \chi(t) = y\}$
which fulfills the necessary conditions of Kramers' Law (\ref{kramerslaw}). From the proof of \ref{transition} we know that for $F_{t,y} = \{\chi : \chi(0) = v^* , \chi(t) = y\}$
\begin{align}
V(t,y) := \inf_{\omega\in F_{t,y}} J[\omega] &= \frac{\theta\left(y - \mu \right)^2}{1- e^{-2\theta t}}
%\\ &= \frac{\beta(y- v_0 - 2\frac{\mu}{\beta} + \frac{\mu}{\beta}e^{-\beta t})^2}{1 - e^{-2\beta t}}
\end{align}
$V(y) = \inf_{t>0}V(t,y)$ can be determined by letting $t \to \infty$:
\begin{align}
V(y) = \lim_{t\to\infty}V(t,y) &= \theta(y- \mu)^2
\end{align}
and therefore for $\delta D = \{x,z\}$ the infimum is assumed at
\begin{equation}
V^* = \inf_{y\in\delta D}V(y) = \begin{cases}
V(z) & \mathrm{if: } |x-\mu| \geq |z-\mu|\\
V(x)  & \mathrm{else}\\
\end{cases}
\end{equation}
and the original statement follows.
\end{proof}
\begin{corollary}\label{boundary}
In a neuron model we generally do not care about the passage through the lower bound $x$ so we can choose it bigger than $z$. Since $x$ then does not appear in the resulting expression it can be chosen arbitrarily large and the result also holds for the one-sided first passage time
\begin{equation}
\mathbb{E}(T^\varepsilon) \approx e^{\frac{\theta}{\varepsilon}(z - \mu)^2}
\end{equation}
\end{corollary}
\begin{figure}
\centering
\centerline{%
\includegraphics[scale=0.25,trim={0 0 10cm 0},clip]{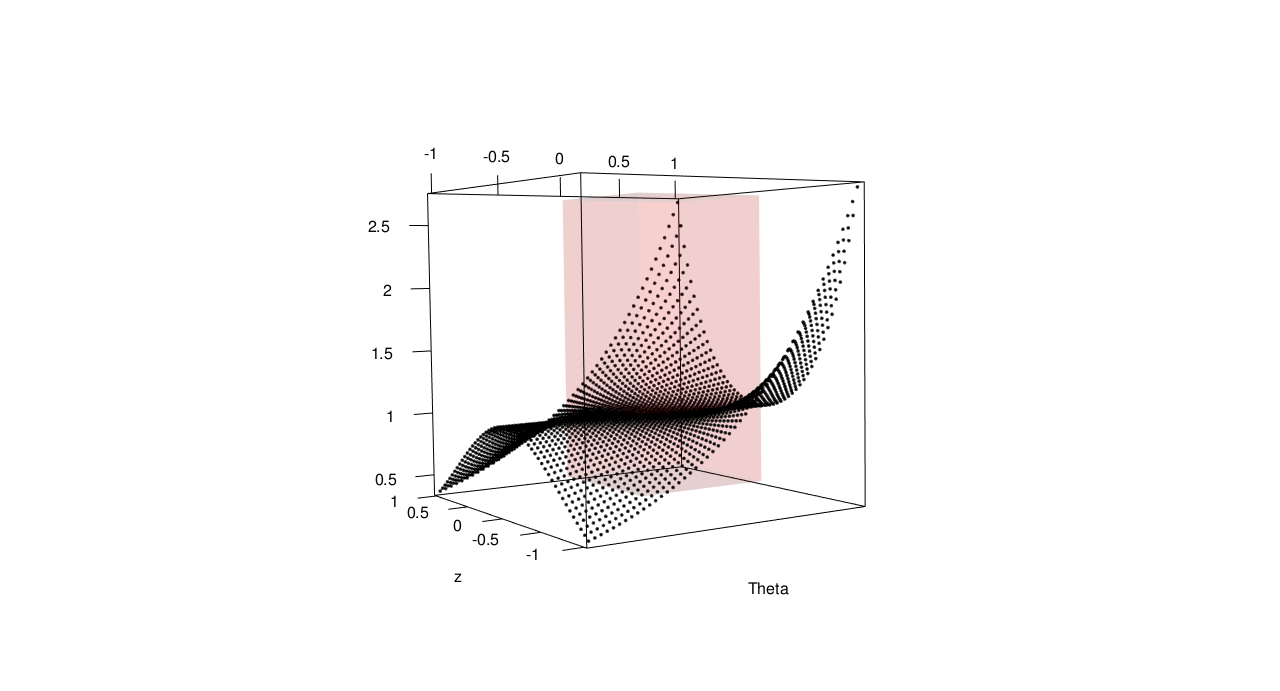}
\includegraphics[scale=0.25,trim={10cm 0 0 0},clip]{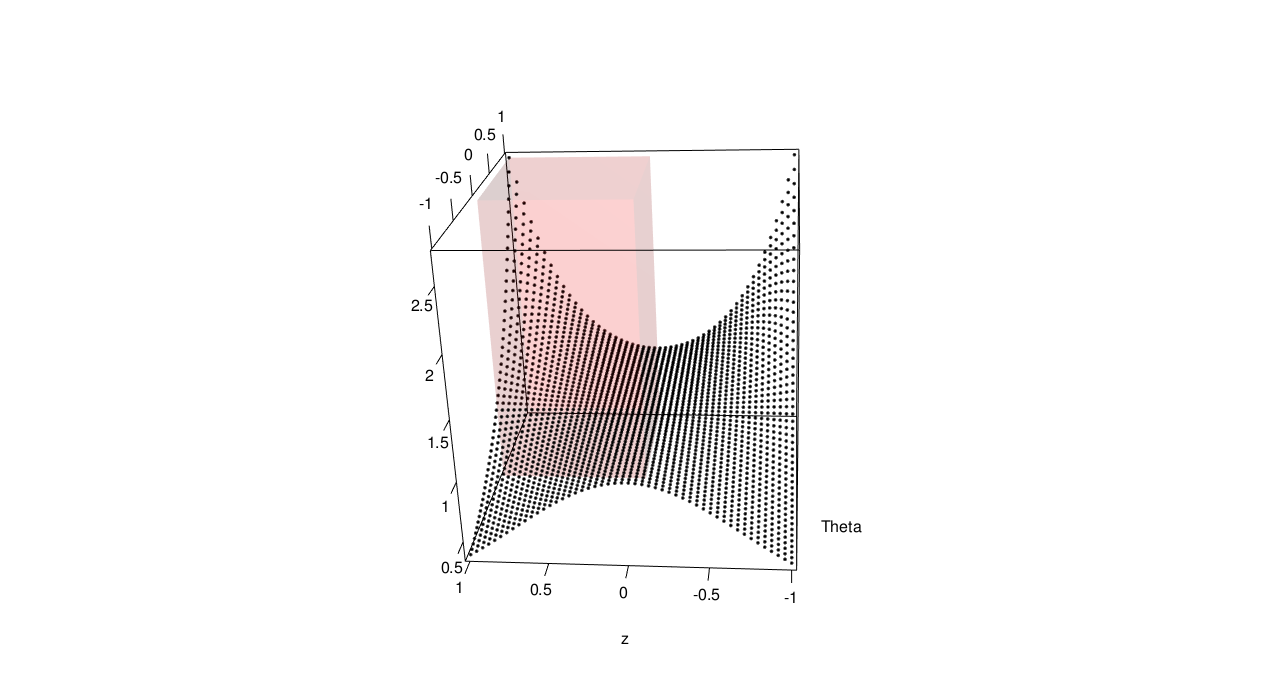}
}%
\caption{A 3D plot of the estimated expected first passage time derived in \ref{boundary} for varying values of the parameters $\theta$ and $z$ assuming $\varepsilon=1, \mu = 0$. Only the portions enclosed by the red box contain plausible parameter values.}
\label{3dplot}
\end{figure}

This is a reasonable result since the FPT should generally be non-negative and it should decrease with a larger decay rate or a larger noise term. It should also be smaller if the threshold is closer to the equilibrium potential. All of these conditions are incorporated into the estimator which I have plotted in Figure \ref{3dplot}. Nonetheless the result can feel rather unintuitive since a smaller domain should reduce the expected first exit time. But since this statement only holds in the small-noise limit we can again apply the intuition of the most likely of all unlikely paths: As soon as the passage through one of the boundaries from $D$ becomes slightly more unlikely than the passage through the other, its relevance in the small-noise limit disappears. Therefore the smaller the value of $\varepsilon$, the closer $L_1$ can be to $L_2$ and choosing $L_1=\infty$ will only improve the rate of convergence of our estimate.

\begin{remark}
When we choose $D = ]v_r - L , v_r + U[$ with $U \geq L$ and set $\mu := v_r + \frac{\mu}{\beta} , \beta = \theta$ then we see that a consistent approximation was derived by \citeauthor{10.2307/1427567} in section 4 of \cite{10.2307/1427567}. It says that in the small noise limit $\varepsilon\to 0$
\begin{equation}\label{gior}
\mathbb{E}(T^\varepsilon) \approx \frac{\sqrt[]{\varepsilon\pi}}{\beta^{\frac{3}{2}}(U-\frac{\mu}{\beta})}e^{\frac{\beta}{\varepsilon}(U - \frac{\mu}{\beta})^2}
\end{equation}
\end{remark}
\begin{remark}
Statements about the asymptotic distribution of $T^\varepsilon$ have been made by \textcite{day1983exponential} where $\frac{T^\varepsilon}{\mathbb{E}(t)}\sim \mathrm{Exp}(1)$ is determined. Therefore $\mathbb{P}(T^\varepsilon \leq t) = \int_0^t \frac{1}{\mathbb{E}(T^\varepsilon)}\exp(-\frac{1}{\mathbb{E}(T^\varepsilon)}x)\mathrm{d}x$ and asymptotically $T^\varepsilon \sim \mathrm{Exp}(e^{-\frac{\theta}{\varepsilon}(z - \mu)})$.
\end{remark}

\section{Alternative derivation via Eyring–Kramers law}
In the case where $\mathbf{b}(V_{\mathrm{m}}(t))$ assumes the form of a negative gradient, i.e. 
\begin{equation}
	\mathbf{b}(X) = - \nabla F(X)
\end{equation}
with $F$ sufficiently smooth and with a finite amount of minima, there is an alternative method to derive a refined estimate for the expected value of the first exit time called Eyring-Kramers law. This law was known as early as \citeyear{eyring1935activated} by \textcite{eyring1935activated} or slightly later by \textcite{kramers1940brownian} but it was only proven rigorously recently by \textcite{bovier2004metastability}. For the one-dimensional case with the above gradient condition it can be stated as
\begin{theorem}\label{eyring}
For $D = ]x , z[ \subset \mathbb{R}, |x-\mu|\geq|z-\mu|$ the expected exit time can be estimated as
\begin{equation}
\mathbb{E}(T_D^\varepsilon) \approx \frac{2\pi}{\sqrt{F''(v^*)|F''(w^*)|}} e^{\frac{2(F(w^*) - F(v^*))}{\varepsilon}}
\end{equation}
where $w^* = \arg\min_{z\in \delta D}F(z)$ and $v^*$ is the unique attractor in $D$ of the deterministic system.
\end{theorem}
Since for 
\begin{equation}
F_{LIF}(V_{\mathrm{m}}(t)) = \frac{\theta}{2}(V_{\mathrm{m}}(t)-\mu)^2
\end{equation} this relation is given for the Stochastic Leaky Integrate and Fire model we can verify the result from the previous section by this method. Note however that in general $\mathbf{b}$ will not be expressible as the gradient of a function $F$\footnote{Both, \textcite{fitzhugh1955mathematical}'s $\mathbf{b}(\colvec{2}{X}{Y}) = \colvec{2}{\frac{1}{\varepsilon}(X - X^3 - Y - s)}{\gamma X - Y + \beta}$ and \textcite{izhikevich2003simple}'s  $\mathbf{b}(\colvec{2}{X}{Y}) = \colvec{2}{0.04X^2 + 5X + 140 - Y + I}{a(bX - Y)}$, do not satisfy the necessary condition of integrability: $\frac{\delta \mathbf{b}_1
}{\delta Y} = \frac{\delta \mathbf{b}_2
}{\delta X} $.} so that the method that uses the rate function is more general.\\
Using theorem \ref{eyring} with $F_{LIF}(V_{\mathfrak{m}}(t))$ and $v^*$ as determined in \ref{attractor} we get
\begin{align}\label{eyring_est}
\theta = F_{LIF}''(V_{\mathrm{m}}(t)) & \quad
z =w^*= \arg\min_{w\in \delta D}F_{LIF}(w) \\
\mathbb{E}(T^\varepsilon_D) &\approx  \frac{2\pi }{\theta}e^{\frac{\theta}{\varepsilon}(z - \mu)^2}
\end{align}
where the second factor dominates for large values of $z, \mu$ and in particular for small values of $\varepsilon$.
\section{Possible application in parameter estimation}
The previous derivation can be useful in and of itself by giving a demonstration for how the expected first passage time (EFPT) can be derived directly from the surface form of the corresponding SDE. In this section I am exploring how the results might be used for estimating parameters of a model that incorporates the OU process and that derives further quantities of interest from it.

For example we can model a neuron by selecting a deterministic process $\mu(t)$ as the input, denoting the cell membrane as $V_\mathrm{m}(t)$ and the random spiking output per unit of time as $\lambda$. Then we assume that the cell membrane $V_\mathrm{m}(t)$ can be represented as an OU process satisfying
\begin{equation}\label{simpleformulation}
{\mathrm{d}V_{{\mathrm  {m}}}(t)} =-\theta\left(V_{{\mathrm  {m}}}(t) - \mu(t)\right)\mathrm{d}t + \varepsilon \mathrm{d}W_t  \quad V_m(0) = v_1
\end{equation}
with a constant threshold $V_\mathrm{th}$.

To use the results from the previous section we need to assume that the input is constant $\mu(t)=\mu$, so for now we assume that our modeled neuron is isolated and receives no input from other cells\footnote{This does not mean that $\mu$ is equal to zero since even in isolation there are ions moving through the cell membrane through active and passive channels.}. In this situation the parameters of the model can be identified for example by maximum-likelihood estimation from single-unit recordings as presented in \textcite{LANSKY1983247} to obtain $\hat{V}_\mathrm{th}$, $\hat{\theta},  \hat{\mu}$ and $\hat{\epsilon}$. The result from the previous section then says that 
\begin{equation}
\mathbb{E}(T^\varepsilon) \approx e^{\frac{\hat{\theta}}{\hat{\varepsilon}}(\hat{V}_\mathrm{th} - \hat{\mu})^2}
\end{equation}
and that it is reasonable to assume $T^\varepsilon \sim \mathrm{Exp}(\mathbb{E}(T^\varepsilon)^{-1})$. We can use the well-known fact\footnote{See for example \textcite{cooper2005poisson}.} that if the distribution of inter-spiking intervals $T^\varepsilon$ is exponential then the distribution of the numbers of spikes $\lambda\in\mathbb{N}_0$ in one unit of time $t$ is Poisson and the two are connected by 
\begin{equation}
\mathbb{E}(\lambda) = \mathbb{E}(T^\varepsilon)^{-1}
\end{equation}
and therefore $\lambda \sim \mathrm{Poiss}(e^{\frac{\theta}{\varepsilon}(V_\mathrm{th} - \mu)^2})$. This shows that in the small-noise limit it can be a reasonable simplification to model a spike trail as a Poisson process with an appropriately chosen rate parameter. This connection has also been noted by \textcite{stevens1996integrate} who derive a different but also exponential expression for the mean. Furthermore, if we instead assume that $\mu(t)$ is no longer constant (but still smaller than $V_\mathrm{th}$), it would be \textit{a priori} very hard to say how this influences the spiking output $\lambda(t)$. But it might be a promising approach to assume that the derived expression for the EFPT does not critically depend on a \textit{constant} input and that for a varying input $\mu(t)$ we can estimate the EFPT as $\mathbb{E}(\lambda) \approx e^{\frac{\hat{\theta}}{\hat{\varepsilon}}(\hat{V}_\mathrm{th} - \mu(t))^2}$.

We can sample a spike trail with respect to $\mu(t)$ by successively considering small discrete time steps $\Delta \mathrm{t}$ and by generating random uniform numbers $U \in [0,1]$ for each time step $\Delta \mathrm{t}$. In each time step we say that a spike occurred if $U < \Delta \mathrm{t} \times e^{\frac{\hat{\theta}}{\hat{\varepsilon}}(\hat{V}_\mathrm{th} - \mu(t))^2}$. Details on this method can be found in the first chapter of the book by \textcite{dayan2001theoretical}. This is a big computational simplification since it allows us to leave out the membrane potential and instead directly go from input to output. I address the question of how valid this approximation is in the next chapter.

Another interesting question is how much information about the input $\mu$ can be retrieved from knowing the expected length of ISIs. This is a very difficult question\footnote{See for example \textcite{wei2004signal}  or \textcite{aihara2002possible} for more sophisticated methods.} and I will only suggest a possible estimator arising from \ref{kramer_est} without analyzing its properties or its biological plausibility further. If we assume that $\theta, \varepsilon$ and $V_\mathrm{th}$ are constants of the neuron\footnote{At least for the diffusion coefficient $\varepsilon$ this is a very strong assumption.} then we can use the \textit{method of moments} to derive
\begin{equation}\label{memo}
\mu = \sqrt[]{\log(\mathbb{E}(T^\varepsilon))\frac{\varepsilon}{\theta}} + V_\mathrm{th}
\end{equation}

Thus if we have neuronal spiking data we can determine the inter-spiking intervals (ISIs) and then derive the EFPT directly from the data by averaging over the length of the ISIs. By using the maximum-likelihood estimates of the neuronal constants $\hat{\varepsilon},\hat{\theta}$ and $\hat{V}_\mathrm{th}$ we can then use \ref{memo} to derive an estimate for $\mu$. 
% Chapter Template

\chapter{Monte-Carlo Sampling for First-Exit Times} % Main chapter title

\label{Chapter6} % Change X to a consecutive number; for referencing this chapter elsewhere, use \ref{ChapterX}

%----------------------------------------------------------------------------------------
%	SECTION 1
%----------------------------------------------------------------------------------------

The problem of estimating the First-passage time (FPT) from a SDE has some inherent intricacies arising from the stochastic nature of the process \textcite{giraudo1999improved}. The naive approach is to sample a large number of paths from the SDE with a discretization scheme and to determine the first point in time where the simulated discretized trajectories cross a fixed threshold. However, no matter how small the discretization in the approximation scheme is chosen, the fluctuations of the process in between two simulated points can be arbitrarily large and thus bias this naive estimator. These effects become even more pronounced if the parameters of the SDE take on very small or very large values as they do in the case of the OU-model of the membrane potential. A possible method to avoid this bias is presented in \textcite{Giraudo2001} or \textcite{drugowitsch2016fast} but there is (to my knowledge) no satisfying implementation available online. Implementing such an algorithm is beyond the scope of this work so I used \texttt{R} to implement the naive method.
\section{Results}
\begin{figure}
\centering
\includegraphics[scale=0.3]{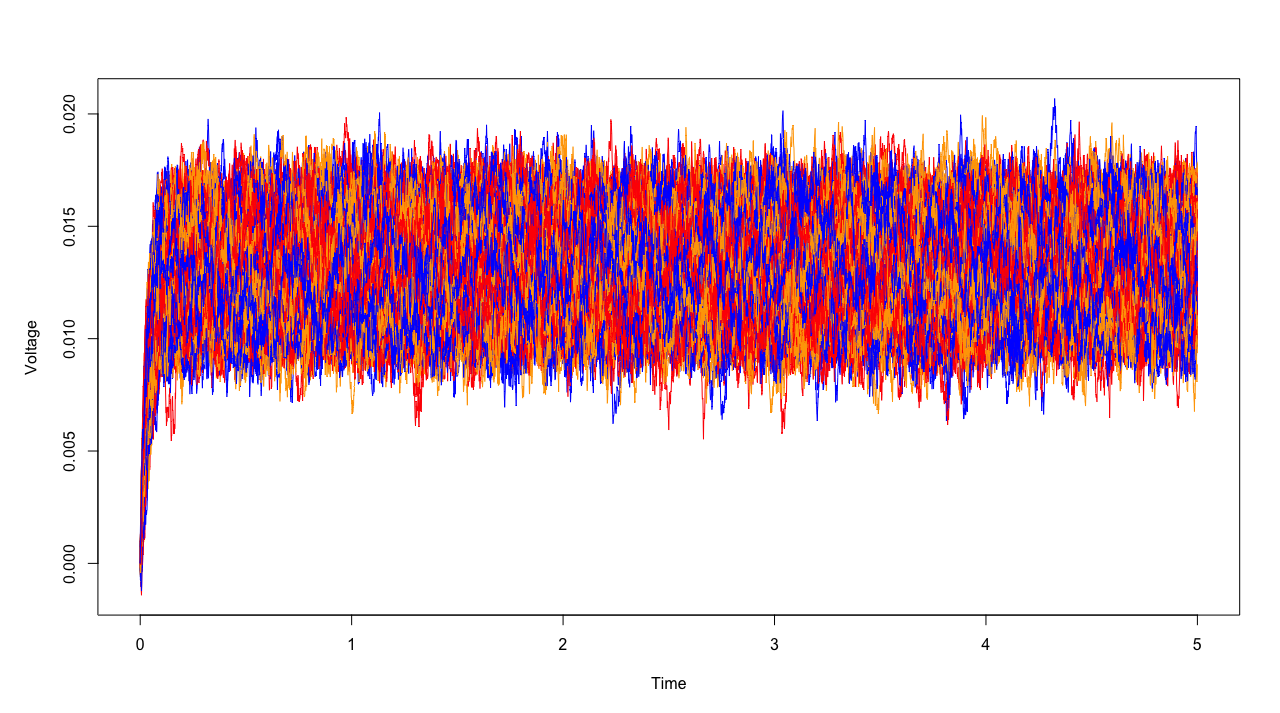}
\caption{500 trajectories of \ref{simpleformulation} simulated with the third-order Runge-Kutta scheme with $\Delta t = 0.0625$ parametrized according to \textcite{lansky2006parameters}.}
\label{500trajfig}
\end{figure}
For the simulation with \texttt{R} (Appendix \ref{AppendixA}) I used the \texttt{Sim.DiffProc} package by \textcite{simdiffproc} which implements several relatively robust functions for the simulation of random diffusion processes and the determination of the first-passage time through a constant boundary. The discretization scheme used for the simulation is the third-order Runge-Kutta scheme presented by \textcite{tocino2002runge}. 500 sample trajectories generated by this method can be seen in \ref{500trajfig}.
\begin{figure}
\centering
\centerline{%
\includegraphics[scale=0.3]{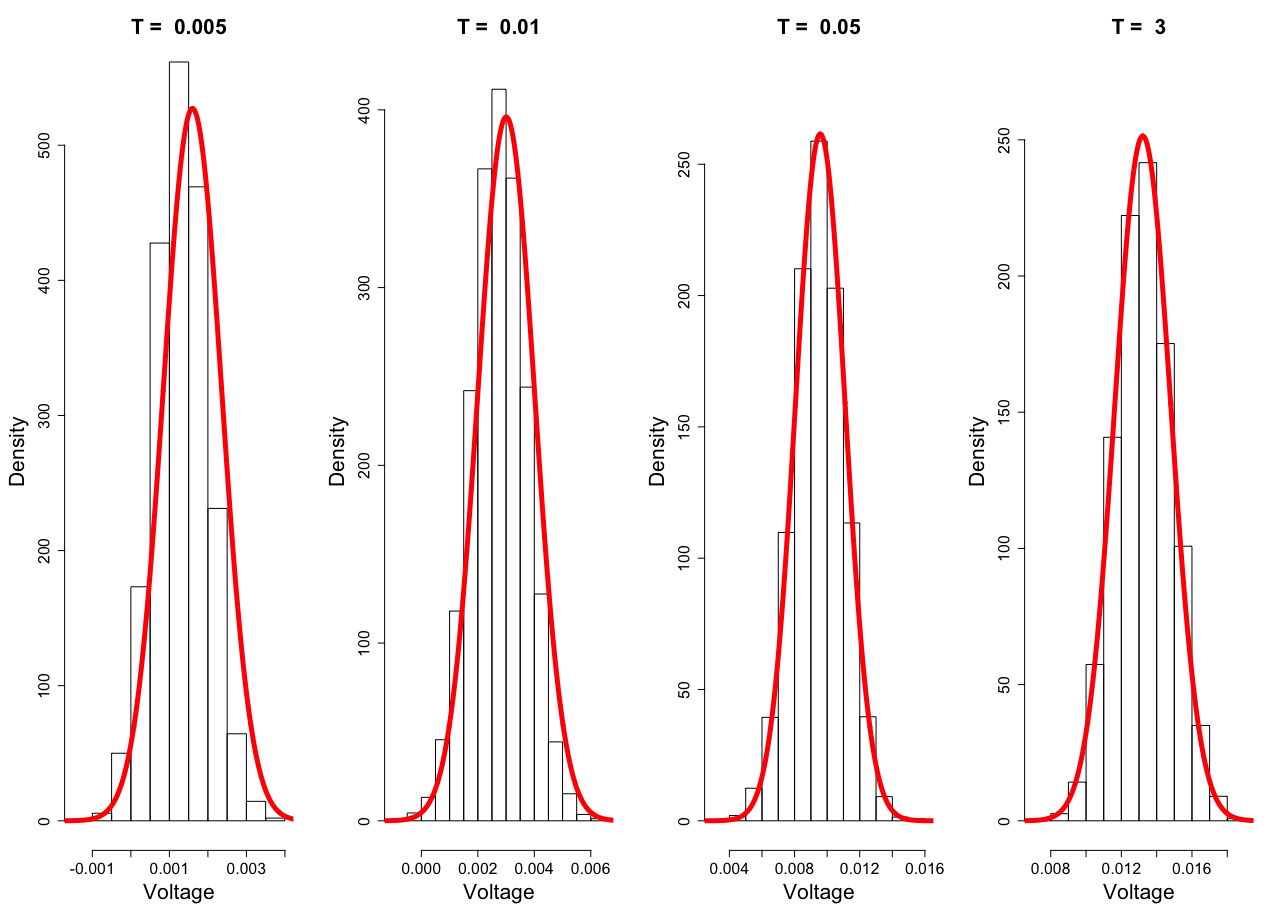}
\includegraphics[scale=0.3]{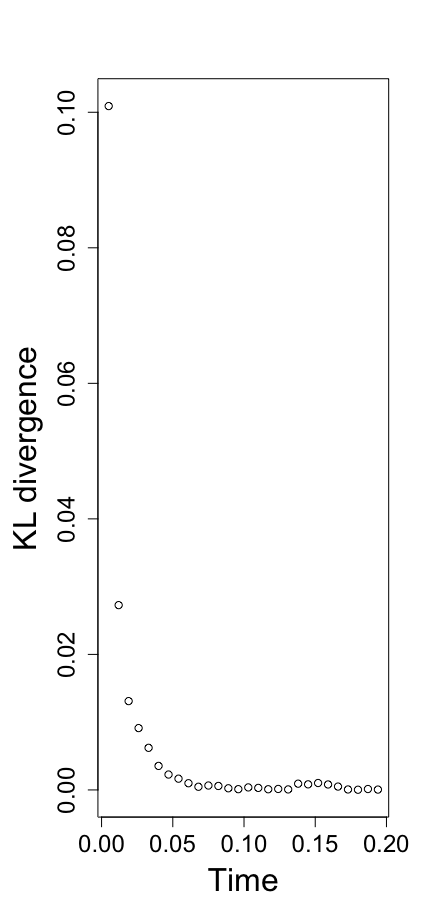}
}%
\caption{Left: Histogram of the distribution of the voltages at different time points t. The predicted transition densities as determined by \ref{transition} is plotted in red. Right: Estimated Kullback-Leibler Divergence of the simulated transition distribution to the theoretical distribution, according to the KL formula for two normal distributions $\mathcal{N}_1,\mathcal{N}_2: \mathrm{KL}(\mathcal{N}_1||\mathcal{N}_2)=\log(\frac{\sigma_2}{\sigma_1}) + \frac{\sigma_1^2 + (\mu_1-\mu_2)^2}{2\sigma_2^2} - \frac{1}{2}$}
\label{figTransHist}
\end{figure}
In \ref{figTransHist} the histogram plot of the membrane potential at ascending points in time is displayed next to a plot of the KL-divergence to the theoretically predicted normal distribution at different points in time. We see that close to $t=0$ the simulation does not match the predictions perfectly but already at $t = 0.05$ the fit is very good. This is probably less a fact about the accuracy of the expression derived in \ref{transition} and more about the accuracy of the \texttt{Sim.DiffProc} package.

\begin{figure}
\centering
\centerline{%
\includegraphics[scale=0.3]{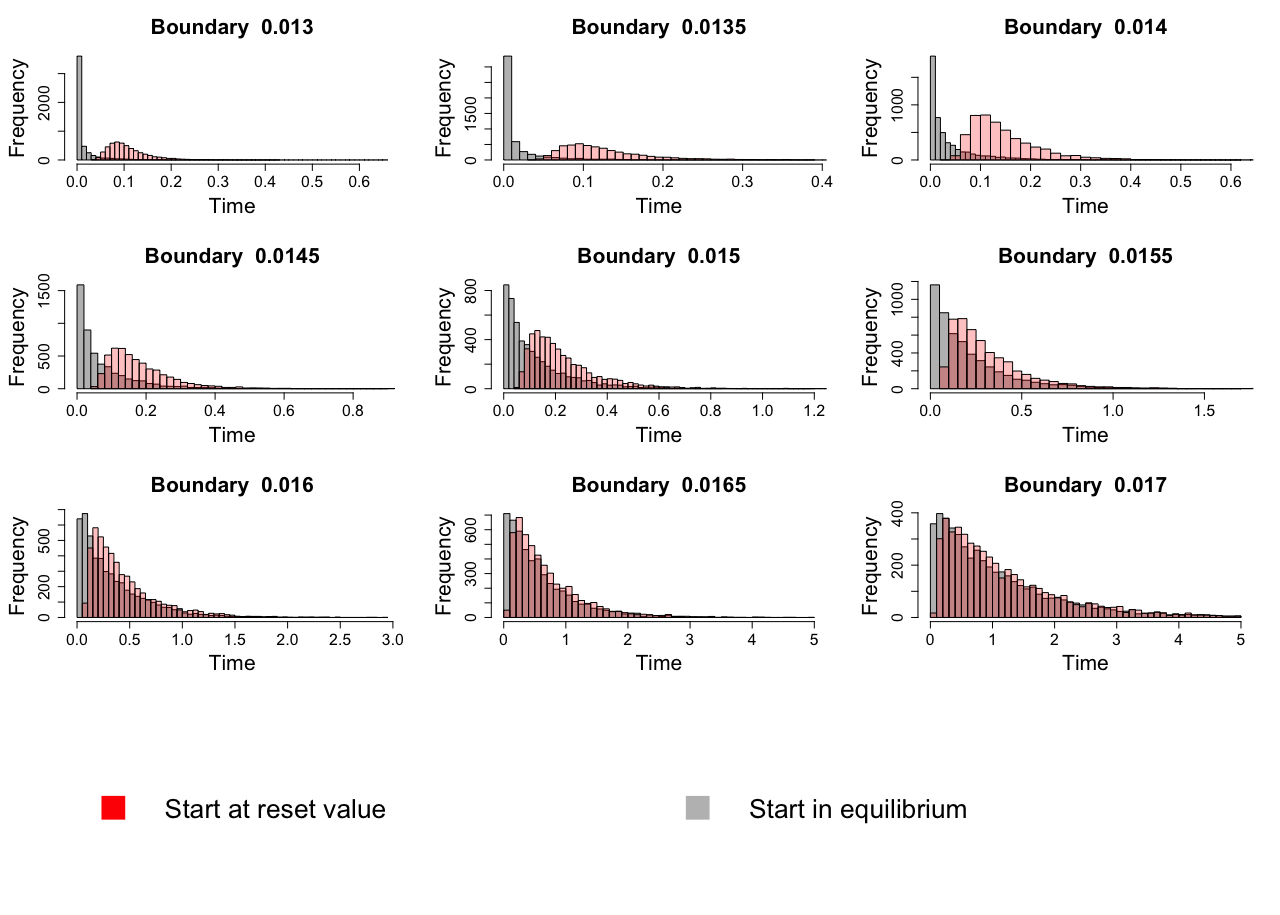}
\includegraphics[scale=0.3]{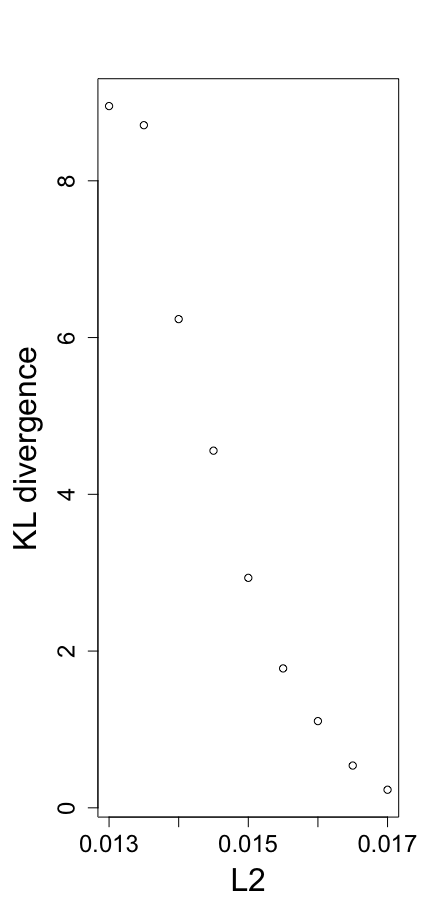}
}%
\caption{Left: Histograms of FPT of two different processes. The process characterized by the initial condition $V_\mathrm{m}(0)=\mu$ is displayed in grey and the process with $V_\mathrm{m}(0)=0$ is displayed in red. Right: Estimated KL divergence between the two samples calculated with the kNN algorithm proposed by \textcite{boltz2007knn} implemented in the \texttt{FNN} package by \textcite{beygelzimer2015package}.}
\label{histogramFPT}
\end{figure}

To determine estimates for the first-passage time I increased the sample size to 5000 with $\Delta t = 0.001$ to diminish the above-mentioned problems arising from fluctuations in between discretization steps. I also considered two different initial conditions of \ref{simpleformulation}: $V_\mathrm{m}(0)=0$, the reset value of the potential after the generation of an action potential, and $V_\mathrm{m}(0)=\mu$, the equilibrium potential. The resulting FPTs are displayed as histograms in \ref{histogramFPT}. For small values of the constant boundary the two conditions differ, but for higher values the dependence on the initial value decreases as the two distributions approach one another.

%% TEST FOR EXPONENTIAL DISTRIBUTION ????

\begin{figure}
\centering
\centerline{%
\includegraphics[scale=0.35]{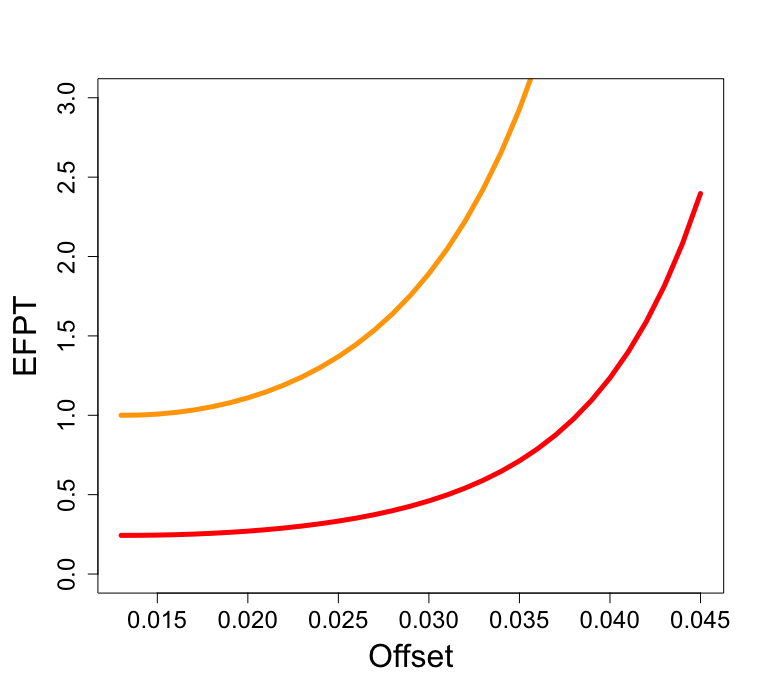}
\includegraphics[scale=0.35]{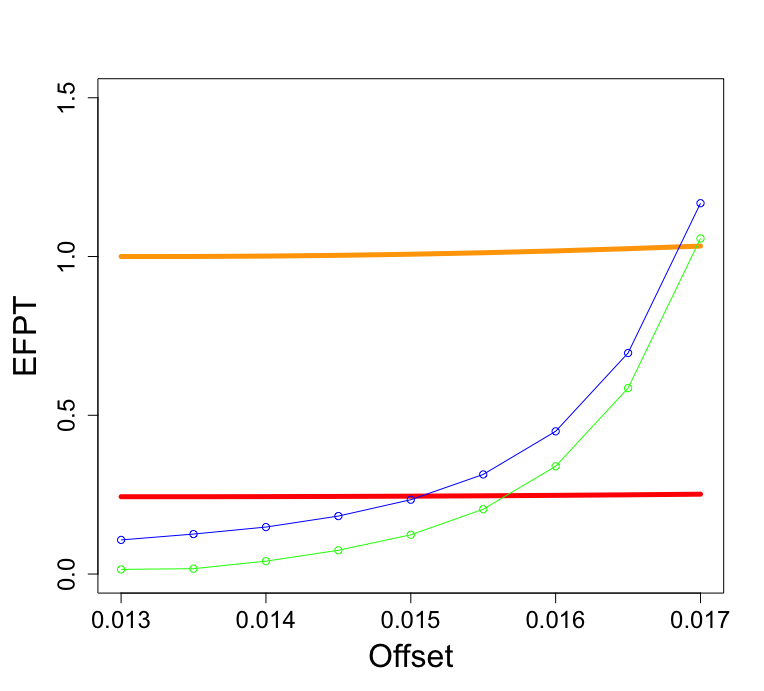}
}%
\caption{Left: The Eyring-Kramer estimate \ref{eyring_est} predicts these expected first-passage times for a given offset from the initial depolarization. Right: Expected first-passage time vs. offset from the initial depolarization $v_0$ after a spike. Naively simulated (green: start at $V_\mathrm{m}(0)=0$, blue: start at $V_\mathrm{m}(0)=\mu$), predicted by Kramers' law \ref{kramer_est} (yellow) and predicted by the Eyring-Kramer theorem \ref{eyring_est} (red).}
\label{figEyring}
\end{figure}
The expression derived in \ref{eyring_est} predicts that if we increase $L_2 > \mu$, 
we will observe an exponential increase in the EFPT of the process. The prediction for fixed values of the other parameters is plotted in \ref{figEyring}. In the same plot the theoretical predictions for this range of offsets is also plotted. From the figure we see that the exponential growth of the simulated EFPT sets in much earlier than the theoretical estimation would predict. To examine where the discrepancy stems from I further analyze the shape of the estimated EFPTs.
\begin{figure}
\centering
\centerline{%
\includegraphics[scale=0.21]{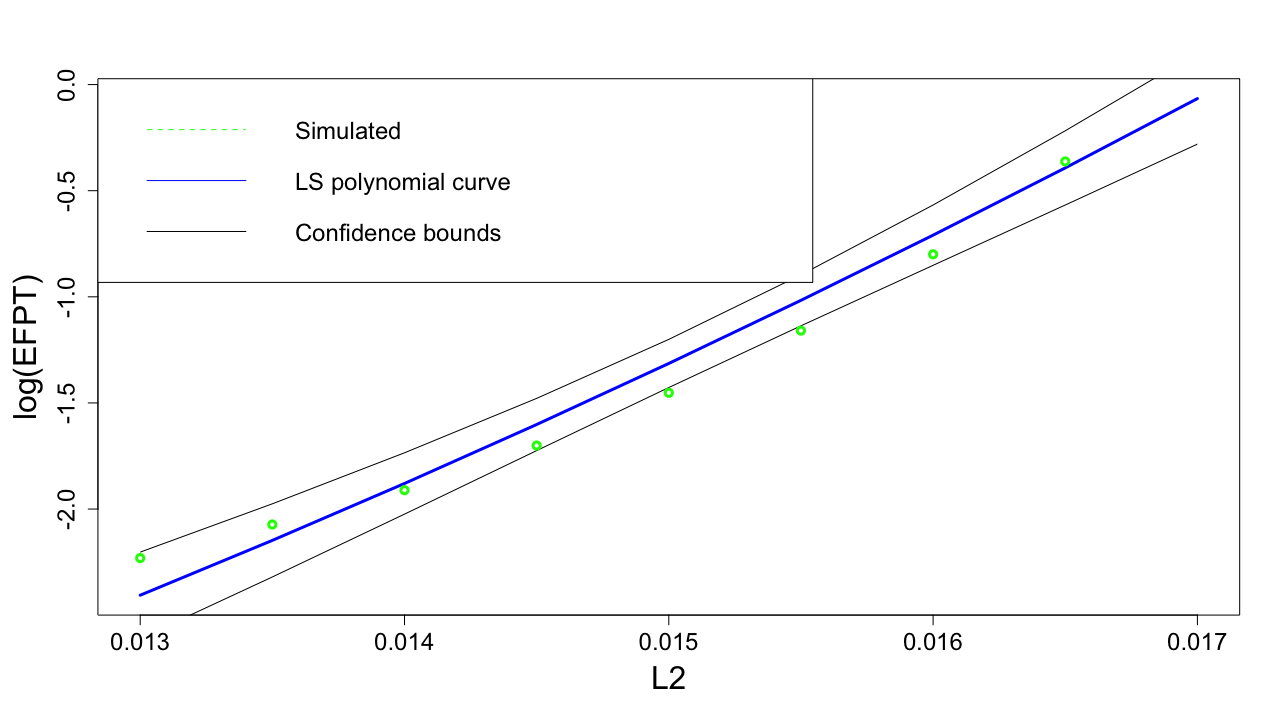}
\includegraphics[scale=0.21]{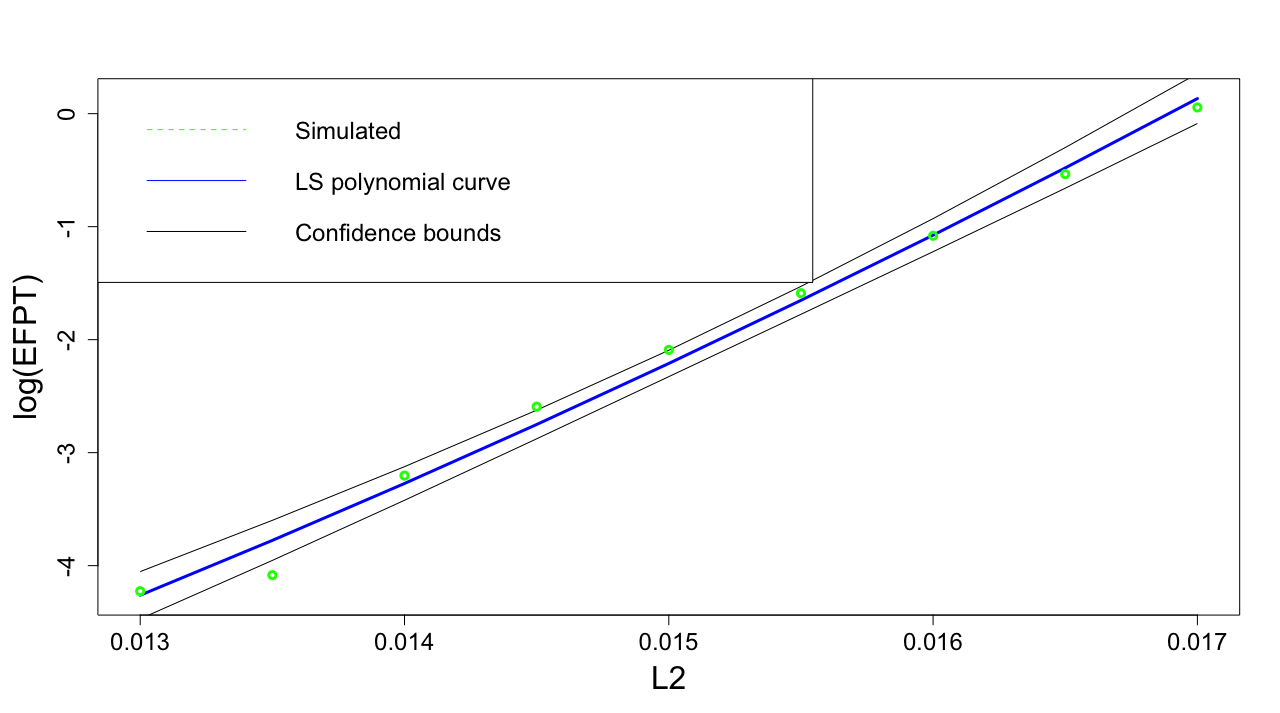}
}%
\caption{Fitting a quadratic model $p(x) = c + ax^2$ with \texttt{R} to the logarithm of the simulated expected means produced the following 95\% confidence intervals: Initial condition at reset potential (Left) - $I_c = [-6.38 ,   -5.03]$ and $ I_a = [16575.83, 22422.35]$%$I_c=[-1.93, -1.89], I_b=[6.55,6.79]$and $I_a =[1.53,1.78]$ 
, Initial condition at equilibrium (Right) - $I_c=[-11.15  ,  -9.76]$ and
 $I_a =[33616.67, 39634.71]$}
\label{polyfit}
\end{figure}
The resulting curve looks like an exponential and so I fitted a quadratic polynomial model without a linear term to the logarithm of the simulated means which produces a satisfying fit, see \ref{polyfit}. For the system which is initialized in equilibrium the fit is even better than for the one initialized at reset potential. The determined coefficient of the quadratic term is one order of magnitude smaller than the predicted value $\frac{\theta}{\varepsilon} \approx 2263.526$. 
\begin{comment}
Leaving all the shortcomings of the estimate aside it is a nice coincidence that for the critical offset of 13mV, where the biological neuron examined by \textcite{lansky2006parameters} is expected to create an action potential, the theoretical estimate lies at least in the close vicinity of the true value.
\end{comment}

% Chapter Template

\chapter{Conclusion} % Main chapter title

\label{Chapter7} % Change X to a consecutive number; for referencing this chapter elsewhere, use \ref{ChapterX}

%----------------------------------------------------------------------------------------
%	SECTION 1
%----------------------------------------------------------------------------------------
In this bachelor thesis I have applied large deviation theory to a simple stochastic neuron model to reproduce two classical results. Using the Freidlin-Wentzell theorem \ref{fwt} I was able to derive the transition probability and the stationary distribution of an Ornstein-Uhlenbeck process directly from the form of its characterizing stochastic differential equation. By extending this result with Kramers' law I derived a qualitative estimate for the expected first passage time. I also provide an additional derivation for this estimate by exploiting the fact that the drift term can be written as the gradient of a scalar field. Finally I provided a few possible applications of the derived results in parameter estimation: On the one hand for deriving a Poisson approximation for the FPT $T^\varepsilon$ given information about the input $\mu$ and also an estimate for the input $\mu$ given the EFPT $\mathbb{E}(T^\varepsilon)$. Simulating the process revealed that while the transition density was derived very accurately, the expression for the expected first-passage time is numerically not very exact. I have also attempted to provide a number of different perspectives on the theorems in large-deviation in general. These perspectives are far from being comprehensive and providing a full neuroscientific interpretation of large-deviation remains an open problem. Although the application of the theorems is not overly complicated, their proofs include a number of involved arguments and a full understanding of them appears to be necessary for a neuroscientific interpretation.

There is a lot of possible further work after this initial foray. The effects biasing the Kramers estimation must be investigated further, possibly following the approach by \textcite{PhysRevE.48.931} who developed a method to determine the subexponential factors missing in a Kramers type estimate. Besides only estimating the expected value of the FPT, an estimate of other invariants like the variance and further higher moments would be very useful. It might even be possible to determine the moments from the complete asymptotic distribution of $T^\varepsilon$, probably using results akin to the asymptotic results from \textcite{day1983exponential}. The simulations indicate that the asymptotic approximation through an exponential distribution is valid if the process is started in equilibrium and that it fails if the process is started at reset potential and the threshold is relatively low. Possibly the FPT $T^\varepsilon$ should be described as the sum of the exponential Kramers estimate starting in equilibrium plus an additional random variable $T_\mathrm{res}^\mathrm{equ}$ that describes the time required to transition from reset to equilibrium.

Besides extending the framework it would also be of interest how it can be integrated into the general Bayesian framework. The aforementioned connections to the KL divergence and to the Gibbs measure might be central to creating such a connection, see Sanov's theorem and the work by \textcite{fischerlarge}. Unifying the two approaches would surely be beneficial for both fields.

Last and possibly most important is to get a better understanding of the biological interpretation of large deviation results. In statistical mechanics the rate function is interpreted as a quasi-potential of a non-conservative field, see \textcite{richard1995overview}, and this interpretation might translate into the current situation. It is then interesting to look at possible biological mechanisms representing and \textit{minimizing} the quasi-potential. To achieve this it might be instrumental to apply the presented theorems to further models of neuronal dynamics, see for example the models by \textcite{hodgkin1952quantitative} or \textcite{izhikevich2003simple}. 

%----------------------------------------------------------------------------------------
%	THESIS CONTENT - APPENDICES
%----------------------------------------------------------------------------------------

\appendix % Cue to tell LaTeX that the following "chapters" are Appendices

% Include the appendices of the thesis as separate files from the Appendices folder
% Uncomment the lines as you write the Appendices

%\include{Appendices/AppendixA}
%\include{Appendices/AppendixB}
%\include{Appendices/AppendixC}

%----------------------------------------------------------------------------------------
%	BIBLIOGRAPHY
%----------------------------------------------------------------------------------------
\begin{multicols}{2}
\printbibliography[heading=bibintoc]
\end{multicols}

%----------------------------------------------------------------------------------------
%----------------------------------------------------------------------------------------
%	DECLARATION PAGE
%----------------------------------------------------------------------------------------

\begin{declaration}
\addchaptertocentry{\authorshipname}

\begin{itemize}
\item[] I hereby certify that the work presented here is, to the best of my knowledge and belief, original and the result of my own investigations, except as acknowledged, and has not been submitted, either in part or whole, for a degree at this or any other university.
\end{itemize} 
\hfill\\
\noindent Signed:\\
\rule[0.5em]{25em}{0.5pt} % This prints a line for the signature
 
\noindent Date:\\
\rule[0.5em]{25em}{0.5pt} % This prints a line to write the date
\end{declaration}
%\newpage
\cleardoublepage

\end{document}